\documentclass[11pt]{article}
\usepackage{fullpage}
\usepackage{bbm}
\usepackage{url}
\usepackage{booktabs}
\usepackage{amsmath,amssymb}
\usepackage{amsthm}
\usepackage{nicefrac}
\usepackage{microtype}
\usepackage{calc}
\usepackage{eqparbox}
\usepackage{enumerate}
\usepackage[usenames,dvipsnames]{xcolor}
\usepackage[colorlinks,citecolor=blue,linkcolor=BrickRed]{hyperref}
\usepackage{algorithm}
\usepackage{graphicx}
\usepackage{subcaption}
\usepackage{appendix}
\usepackage[noend]{algpseudocode}
\usepackage[colorinlistoftodos]{todonotes}
\usepackage{xr}
\usepackage{array}
\usepackage{xspace}
\usepackage[capitalise]{cleveref}
\usepackage{parskip}
\usepackage{mathtools}
\usepackage{tikz}
\usepackage{thm-restate}
\usetikzlibrary{tikzmark}
\usetikzlibrary{calc}

\usepackage{nameref}
\usepackage{etoolbox}

\DeclareMathOperator*{\E}{\mathbb{E}}
\let\poly\relax
\DeclareMathOperator*{\poly}{poly}

\DeclareMathOperator*{\wt}{wt}
\DeclareMathOperator*{\probability}{\Pr}

\newcommand{\vol}{\text{vol}}
\newcommand{\codim}{\text{codim}}
\newcommand{\supp}{\text{supp}}

\renewcommand{\skew}{\textsc{Skew}\xspace}

\newcommand\F{\mathbb{F}}
\newcommand\R{\mathbb{R}}

\newcommand\st{\text{s.t.}\xspace}
\newcommand\eps{\epsilon}

\newcommand\norm[1]{\left\| #1 \right\|}

\newcommand\inorm[1]{\norm {#1 }_\infty}
\newcommand\pnorm[2]{\norm {#1 }_{#2}}

\newcommand\restr[2]{\ensuremath{\left.#1\right|_{#2}}}

\newcommand\expect[1]{\E\left[ #1 \right]}
\newcommand\expectarg[2]{\E_{#2}\left[ #1 \right]}
\newcommand\prob[1]{\probability\left( #1 \right)}

\newcommand\hatf{\widehat f}
\newcommand\hatpsi{\widehat{\psi}}
\newcommand\hatmu{\widehat{\mu}}
\newcommand\hatomega{\widehat{\omega}}
\newcommand\fsr{\textsc{FindSkew$^+$}\xspace}
\newcommand\fsn{\textsc{FindSkew$^-$}\xspace}
\newcommand\ffc{\textsc{FindFourierCoefficients}\xspace}
\newcommand\pc{\textsc{PreprocessCoefficients}\xspace}
\newcommand\dsc{\textsc{DeduceSubcubeCoefficients}\xspace}

\newcommand\Tr{\mathrm{Tribes}}

\def\C{\mathcal{C}}

\def\x{\mathbf{x}}
\def\y{\mathbf{y}}
\def\z{\mathbf{z}}

\newcommand{\TODO}[1]{}

\newcommand{\remove}[1]{}

\newcommand*{\ip}[1]{\left\langle #1 \right\rangle}

\newcommand{\defeq}{:=}

\newtheorem*{Conj*}{Conjecture}

\let\oldproofname=\proofname
\renewcommand{\proofname}{\rm\bf{\oldproofname:}}

\newcommand{\psc}{\mathrm{PIDScore}}

\newcommand{\D}{\mathcal{D}}

\newcommand{\eat}[1]{}

\newcommand{\rgta}{\ensuremath{\rightarrow}}

\newcommand{\zo}{\{0,1\}}

\newcommand{\pmo}{\ensuremath \{ \pm 1\}}

\theoremstyle{plain}
\newtheorem{theorem}{Theorem}
\newtheorem{lemma}{Lemma}
\newtheorem{fact}{Fact}
\newtheorem{corollary}[lemma]{Corollary}

\newtheorem{problem}{Problem}
\theoremstyle{definition}
\newtheorem{definition}{Definition}

\newlength{\continueindent}
\makeatletter
\newcommand*{\ALG@customparshape}{\parshape 2 \leftmargin \linewidth \dimexpr\ALG@tlm+\continueindent\relax \dimexpr\linewidth+\leftmargin-\ALG@tlm-\continueindent\relax}
\apptocmd{\ALG@beginblock}{\ALG@customparshape}{}{\errmessage{failed to patch}}
\makeatother

\makeatletter
\def\thm@space@setup{%
	\thm@preskip=\parskip \thm@postskip=0pt
}
\makeatother

\errorcontextlines\maxdimen

\newcommand{\ALGtikzmarkcolor}{black}
\newcommand{\ALGtikzmarkextraindent}{4pt}
\newcommand{\ALGtikzmarkverticaloffsetstart}{-.5ex}
\newcommand{\ALGtikzmarkverticaloffsetend}{-.5ex}
\makeatletter
\newcounter{ALG@tikzmark@tempcnta}

\newcommand\ALG@tikzmark@start{%
	\global\let\ALG@tikzmark@last\ALG@tikzmark@starttext%
	\expandafter\edef\csname ALG@tikzmark@\theALG@nested\endcsname{\theALG@tikzmark@tempcnta}%
	\tikzmark{ALG@tikzmark@start@\csname ALG@tikzmark@\theALG@nested\endcsname}%
	\addtocounter{ALG@tikzmark@tempcnta}{1}%
}

\def\ALG@tikzmark@starttext{start}
\newcommand\ALG@tikzmark@end{%
	\ifx\ALG@tikzmark@last\ALG@tikzmark@starttext
	\else
	\tikzmark{ALG@tikzmark@end@\csname ALG@tikzmark@\theALG@nested\endcsname}%
	\tikz[overlay,remember picture] \draw[\ALGtikzmarkcolor] let \p{S}=($(pic cs:ALG@tikzmark@start@\csname ALG@tikzmark@\theALG@nested\endcsname)+(\ALGtikzmarkextraindent,\ALGtikzmarkverticaloffsetstart)$), \p{E}=($(pic cs:ALG@tikzmark@end@\csname ALG@tikzmark@\theALG@nested\endcsname)+(\ALGtikzmarkextraindent,\ALGtikzmarkverticaloffsetend)$) in (\x{S},\y{S})--(\x{S},\y{E});%
	\fi
	\gdef\ALG@tikzmark@last{end}%
}

\apptocmd{\ALG@beginblock}{\ALG@tikzmark@start}{}{\errmessage{failed to patch}}
\pretocmd{\ALG@endblock}{\ALG@tikzmark@end}{}{\errmessage{failed to patch}}
\makeatother

\algdef{SE}[DOWHILE]{Do}{doWhile}{\algorithmicdo}[1]{\algorithmicwhile\ #1}%

\makeatletter
\newcounter{algorithmicH}
\let\oldalgorithmic\algorithmic
\renewcommand{\algorithmic}{%
	\stepcounter{algorithmicH}
	\oldalgorithmic}
\renewcommand{\theHALG@line}{ALG@line.\thealgorithmicH.\arabic{ALG@line}}
\makeatother

\title{Finding Skewed Subcubes Under a Distribution}
\author{
  Parikshit Gopalan\\
  VMware Research\\
  {\tt pgopalan@vmware.com}\and
  Roie Levin \\
  Carnegie Mellon University\\
  {\tt roiel@cs.cmu.edu}\thanks{Work done while an intern at VMware Research.} \and
  Udi Wieder\\
  VMware Research\\
  {\tt uwieder@vmware.com}
}

\date{}
\begin{document}
\thispagestyle{empty}
\maketitle

\begin{abstract}

Say that we are given samples from a distribution $\psi$ over an
$n$-dimensional space. We expect or desire $\psi$ to
behave like a product distribution (or a $k$-wise independent
distribution over its marginals for small $k$). We propose the problem of
enumerating/list-decoding all large subcubes where the distribution $\psi$ deviates
markedly from what we expect; we refer to such subcubes as skewed
subcubes. Skewed subcubes are certificates of dependencies between
small subsets of variables in $\psi$. We motivate this problem by
showing that it arises naturally in the context of algorithmic
fairness and anomaly detection.

In this work we focus on the special but important case where the
space is the Boolean hypercube, and the expected marginals are uniform.
We show that the obvious definition of skewed subcubes can lead to
intractable list sizes, and propose a better definition of a minimal
skewed subcube, which are subcubes whose skew cannot be attributed
to a larger subcube that contains it. Our main technical
contribution is a list-size bound for this definition and an
algorithm to efficiently find all such subcubes. Both the bound and
the algorithm rely on Fourier-analytic techniques, especially the
powerful hypercontractive inequality.

On the lower bounds side, we show that finding skewed subcubes is as
hard as the sparse noisy parity problem, and hence our algorithms
cannot be improved on substantially without a breakthrough on this
problem which is believed to be intractable. Motivated by this, we
study alternate models allowing query access to $\psi$ where finding
skewed subcubes might be easier.

\end{abstract}

\newpage

\section{Introduction}

Assume that we observe samples from a distribution $\psi$ over
points in $n$-dimensional space $\D^n$. Our prior belief is that each attribute has a marginal
distribution $\mu_i$ and that the various
attributes are nearly independent (or at least $k$-wise independent
for small $k$), hence $\psi$ is close to the product distribution $\mu
=\prod_i \mu_i$. Our goal is to find significant deviations between
our hypothesis $\mu$ and the observed distribution $\psi$, manifested
as significant dependencies between small sets of
variables. The distribution $\mu$ might represent either a prior model
for $\psi$, or it might be represent a target distribution that we
wish $\psi$ to be close to.  This problem arises naturally in several machine learning applications as we detail in Section~\ref{sec:motivation}, but first we formulate the problem with more detail.

To formulate a precise statement, we first define the notion of subcubes.
Assume that $\D$ is ordered and bounded, the two canonical examples are
$\D^n = \zo^n$ and $\D^n = [0,1]^n$. Let $K
\subseteq [n]$ be a set of $k$ coordinates. For $x \in \D^n$, $x_K$ denotes
the projection of $x$ onto coordinates in $K$. For each $j \in K$,
let $I_j \subsetneq \D$ be an interval in $\D$.
We call the  set of points $C = \{x \in \D^n: x_K \in \prod_{j \in K}I_j\}$
a subcube of codimension $k$. We have
\[ \mu(C) \defeq \Pr_{\x \sim \mu}[\x \in C] = \prod_{j \in K}\Pr_{\x_j \sim \mu_j}[\x_j \in
  I_j] =   \prod_{j \in K} \mu_{j}(I_j).\]
If we similarly define $\psi(C) \defeq \Pr_{\x \sim \psi}[\x \in C]$,
then our goal is to find subcubes such that $\left| \mu(C) - \psi(C)
\right| \geq \gamma$. Motivated by our applications, we add two more
desiderata to our problem formulation (that will be justified
shortly): we restrict to large subcubes, and we want algorithms that
enumerate all subcubes that satisfy our conditions.

One way to restrict to large subcubes is to only consider subcubes
with $\mu(C) \geq \eta$ for some $\eta \in [0,1]$. Alternately, we
could bound the codimension by $k$. The advantage of the latter is
that we only need that $\mu$ is $k$-wise independent for the equality
$\mu(C) = \prod_{j \in K} \mu_j(I_j)$ to hold. In the discrete case $\D^n =
\zo^n$, the two notions coincide since $\mu(C) = 2^{-k}$ for subcubes
of codimension $k$.

Rather then phrasing this as an optimization question where the goal
is to find the subcube that maximizes the deviation $\gamma$, our goal
will be to come up with a {\em list-decoding} style algorithm that
enumerates over all subcubes of codimension $k$ such that $\left|
\mu(C) - \psi(C) \right| \geq \gamma$.

In addition to being a natural algorithmic question in its own right,
this problem comes up in recent work in machine learning, on anomaly
detection and fairness.

\subsection{Motivation}\label{sec:motivation}
\subsection*{Fairness in Machine Learning}

Assume there is a base population $P$ of individuals, each described by $n$
attributes. We naturally view $P$ as inducing a distribution $\mu$ on
the attribute space $\D^n$. Suppose that small subsets of the attributes are nearly
independent, so that $\mu$ is close to being $k$-wise
independent for some $k$ which is small compared to $n$. We are given a distribution $\psi$
over this population. Our goal is to discover significant biases in the distribution
that are not present in the original population $P$. For instance the
population $P$ might be the set of students that apply to a
university, and $\psi$ might represent the set of successful
applicants. Or $P$ might be the training data for a machine
learning algorithm while $\psi$ represents the misclassified inputs.
The latter setting has received a fair amount of
attention in the context of algorithmic bias and fairness in Machine learning, where the
most commonly studied notion is that of intersectionality bias \cite{cabrera2019fairvis}: we are
interested in biases where we restrict the values of some small subset
of attributes, which are typically discrete. See for instance a recent study
showing that facial recognition software has higher error rates for
women of color \cite{BuolamwiniG18}. Our motivation for considering
subcubes is that it captures intersectionality in the discrete setting.

Enumerating over all subcubes is more appropriate than optimization
in this setting since not all intersectionalities might be equally important.
The fact that college applications submitted during certain
days of the week are less likely to be accepted might not be as
significant as the fact that certain zipcodes are less likely to be
accepted; even if the deviation is lower in the latter case.
We ask for algorithms that enumerate over all biased subcubes and leave it to subject experts to
decide how interesting these are, just as in list-decoding we do not
worry about how the receiver chooses from the list of possible
codewords returned by the decoder. Another reason to favor enumeration
is that in real-world datasets, we may not expect $\psi$ to be truly
$k$-wise independent; we might expect correlations
between certain sets of attributes. But even so, an
exhaustive list of significant correlations might lead us to discover
interesting new properties of the distribution and refine our model
for $\psi$. The restriction to subcubes of bounded codimension is
natural since intersectionalities of few attributes are more interesting.

\subsection*{Anomaly Detection}

Anomaly detection is a ubiquitous unsupervised learning problem \cite{ChandolaBK09}.
Isolation based methods for anomaly detection have proven to be
extremely effective in practice \cite{iForest, emmott2013systematic,
  GuhaMRS16}. Building on this, the recent work of \cite{GopalanSW19} proposes
an approach to anomaly detection based on a notion called Partial
Identification. It assigns a score denoted $\psc(x, P)$ to each
point $x \in P$ which measures how easy it is to distinguish $x$ from
other points in $P$. They give a heuristic to compute
$\psc(x, P)$, and show that the resulting anomaly detection algorithm
outperforms several popular anomaly detection
methods, across a broad range of benchmarks.

Formally, given a set of points $P \subseteq \D^n$ and a
subcube $C \in \D^n$, define the sparsity of $C$ as
\[ \rho(C) = \frac{\vol(C)}{|C \cap P|} \]
The PIDScore of a point $x \in P$ is the maximum  value of $\rho(C)$
over all subcubes that contain it.
\[ \psc(x, P) = \max_{C \ni x} \rho(C). \]
Anomalous points are those for which $\psc(x, P) \geq t$ for some
threshold $t$. Equivalently, it suffices to find all $C$ such that
$\rho(C) \geq t$, and then take all the points contained in them.

To relate this to our problem, let us take $\mu$ to be the uniform
measure over $\D^n$ and $\psi$ to be the measure induced by $P$.
Rescaling $\rho$ by a factor of $|P|/\vol(\D^n)$, we get
\[ \rho'(C) = \frac{\vol(C)}{\vol(\D^n)}\frac{|P|}{|C \cap P|}
  \approx \frac{\mu(C)}{\psi(C)}.\footnotemark\]
  \footnotetext{Actually $|C \cap P|/|P| =
  	\psi(C)$ only in
  	expectation, but since subcubes have small VC dimension, we get
  	tight concentration.}

If we also scale the threshold $t$ by the same factor, then the set of outliers
stays the same. But $\rho'(C) \geq t'$ implies $\psi(C) \leq
\mu(C)/t'$, hence
\[ \frac{\mu(C) - \psi(C)}{\mu(C)} \geq 1 - \frac{1}{t'}.\] Thus this is an
    instance of the problem that we consider, where our goal is to
    find non-empty subcubes that are underrepresented in $\psi$, when
    compared to $\mu$.
    Enumeration over all sparse subcubes is natural
    in this setting, since we wish to list all points with high scores.

\section{Our Results}

In this paper, we focus on the case when $\D^n = \pmo^n$ and
$\mu$ is the uniform distribution. We believe that several of our
techniques apply to more general product distributions.
A subcube of codimension $k$ is
obtained by restricting the values of some subset $K$ of
coordinates. For subcubes $C \subseteq D$ we refer to $C$ as a child
of $D$ and $D$ as a parent of $C$.

As a warm-up we first consider the following problem:
\begin{problem}
  {\bf Finding skewed subcubes.}
  Given sample access to a distribution $\psi$ over $\pmo^n$ and $\gamma \in (0, 2^k
  -1]$ find all subcubes $C$
  with codimension $j \leq k$ such that
  \[ \left|\frac{\Pr_{\x \sim \psi}[\x \in C] - 2^{-j}}{2^{-j}}\right| \geq
  \gamma. \]
\end{problem}

There is a trivial $\tilde O(n^k)$ algorithm that enumerates over all
subcubes. To beat this naive bound, we first need to bound the list-size of the output, or rather a
bound on the number of skewed subcubes. However, we show in
\cref{lem:comb_lb} that there exist distributions where the number of
skewed subcubes is $\Omega((n/k)^k)$, which is not far from the
trivial upper bound.

The proof of  \cref{lem:comb_lb} demonstrates that one source for the abundance of skewed subcubes is that skew is easily inherited by
children from their parents: if a subcube $C$ of codimension $j$ is skewed,
for every choice of $k -j$ additional coordinates, by simple
averaging, there is at least one  restriction that results in a skewed
subcube. So even if we consider the uniform distribution over points
with $x_1 =1$, there are $\Omega(n^{k-1})$ skewed subcubes by this
definition, while really the only interesting subcube is the $x_1 =1$
subcube. Our first contribution is a definition which captures only those subcubes that do not inherit their skew from a parent.

\begin{problem}
  {\bf Finding minimal skewed subcubes.}
  Given sample access to a distribution $\psi$ over $\pmo^n$, $\gamma \in (0, 2^k
  -1]$ and $\eps \in (0,1)$ find all subcubes $C$
  with codimension $j \leq k$ such that
  \[ \left|\frac{\Pr_{\x \sim \psi}[\x \in C] - 2^{-j}}{2^{-j}}\right| \geq
  \gamma. \]
  and for every parent $C' \supsetneq C$ of codimension $i$,
  \[ \left|\frac{\Pr_{\x \sim \psi}[\x \in C] - 2^{-i}}{2^{-i}}\right| \leq
  \gamma(1 - \eps). \footnotemark\]
\footnotetext{The formal definition of
	minimal skewed subcubes (Definition \ref{def:min_skew}) is a little more
	involved, we only care about those parents of $C$ which are skewed the
	same way as $C$.}
\end{problem}

We refer to such a subcube as a $(\gamma, \epsilon)$-minimal skewed
subcube. This notion is motivated by our applications: if we already know
that $\Pr_{\x \sim \psi}[(\x_1 = 1) \wedge (\x_2 = 1)] = 3/4$ (rather
than $1/4$), then knowing that
$\Pr_{\x \sim \psi}[(\x_1 = 1) \wedge (\x_2 = 1) \wedge (\x_3 =1)] =
3/8$ should not surprise us, given our prior.

A natural question to ask is whether focusing in minimal skewed
subcubes suffices to make the problem (or at least the list size) more
tractable. Our second contribution is a bound on the number of minimal
skewed subcubes which is independent of the dimension $n$. Instead
we have a dependence on the max norm of the probability
distribution defined below.

\begin{def}
  \label{def:max-norm}
  Given a distribution $\psi$ on $\pmo^n$, let $\inorm{\psi} \defeq
  2^n\cdot \max_{x \in \pmo^n}\Pr_{\x \sim \psi}[\x =x]$.
\end{def}
The parameter $\inorm{\psi}$  lies in the range $[1, 2^n]$ and is a
measure of how well-spread the distribution is. It is referred to as
the smoothness of a distribution in the literature on boosting, and is
closely related to min-entropy. The uniform distribution has $t=1$,
whereas $t = 2^n$ when the entire distribution is concentrated on a
single point.

\begin{theorem}\label{thm:main1}
  For a distribution $\psi$ on $\pmo^n$, the number of
  $(\gamma, \eps)$-minimal skewed subcubes of codimension at most $k$
  is bounded by  $k^{O(k)}\left(\ln(e
  \inorm{\psi})\poly(1/\eps,1/\gamma)\right)^k$.
\end{theorem}

For constant $\eps, \gamma$, the asymptotic dependence on $n$ is never worse than
$O(n^k)$, which happens when $\psi$ is concentrated on a point.
But when $\inorm{\psi} = O(1)$, the above bound is $O_k(1)$
and when  $\inorm{\psi} = \poly(n)$, the bound is
$O_k(\ln(n)^k)$ improving substantially over the  $O(n^k)$ bound.

There are two key elements in the proof of \cref{thm:main1}. We first use a
novel Fourier based algorithm to reduce the problem to that of
finding large, low-degree Fourier coefficients in a series of
restrictions of the distribution $\psi$ to various subcubes. We then
use the powerful hypercontractive inequality to bound the number of
such coefficients in any distribution in terms of $\inorm{\psi}$.
This latter bound generalizes the level-$k$ inequalities for
indicators of small sets in the Boolean hypercube \cite[Chapter
  9]{o2014analysis}, and the proof follows similar lines.
We also construct distributions showing that for various values of
$\eps, \gamma$, the dependency of $(\ln(\inorm{\psi}))^k$ is
optimal. The distributions are constructed using the
Tribes function and BCH codes.

We now turn to the algorithmic problem of finding the list of minimal
skewed subcubes. We observe that even when the list-size is constant,
there is a significant algorithmic barrier to a $n^{o(k)}$ algorithm,
namely the $k$-sparse noisy parity problem \cite{FeldmanGKP09,
  Valiant15}. In this problem, we are given points $x$ and labels $y$
which are the XOR of some $k$-subset $S$ with random noise of rate
$\eta$ added. There is a simple reduction from this problem to finding
skewed subcubes, if we consider the distribution of $(x,y) \in
\pmo^{n+1}$ the only skewed subcubes involve the coordinates $S \cup
\{n +1\}$.

\begin{restatable}{theorem}{noisyparity}
	\label{thm:noisy-parity}
  	For $\eta \in (0, 1/2)$, an algorithm that given a distribution $\psi$
        and $k$ can find a $(1- 2\eta, 1)$-minimal skewed subcube of
        co-dimension $k$ in time $T(n,k, \eta)$ can be used to solve
        the $k$-sparse noisy parity problem with noise rate $\eta$ in
        time $T(n,k, \eta)$.
\end{restatable}

Given this reduction, there are two {\em lower bounds} on the running time
of any list-decoder: the list-size given in Theorem \ref{thm:main1}, and the running time of the best
known algorithm for the $k$-sparse noisy parity problem, which is
$O(n^{0.8k})$ due to \cite{Valiant15}. We give an
algorithm that nearly gets the sum of these two bounds.

\begin{theorem} \label{thm:intro_alg}
	For any measure $\psi$ on $\pmo^n$, integer $k \leq n$, and
        parameters $0 \leq \gamma \leq 2^k - 1$ and $0 \leq \eps \leq
        1$, there are algorithms that return all $(\gamma,
        \eps)$-minimal skewed subcubes of codimension at most $k$ in
        time
        \[\tilde O \left(n^{0.8 k}\right)+  
        \tilde O \left( n^{k/3}\right) \cdot \frac{k^{O(k)}}{(\eps \gamma)^{4/\lambda + 2k}} \left(\ln \left(\frac{e \inorm{\psi}}{\eps \gamma}\right)\right)^k \]
\end{theorem}

Finally, to circumvent the noisy parity problem, we consider stronger
models where we have query access to the distribution $\psi$: in
addition to random samples, we can also query the value of $\psi(x)$
for any $x \in \pmo^n$. The noisy parity problem becomes trivial to solve once one has query
access. In this model, we are able to get an algorithm whose running time is
is $\poly(n, \inorm{\psi})$. Thus when $\inorm{\psi} < n^{\alpha k}$ for some $\alpha > 0$,
this improves over the trivial algorithm. We show some
dependence on $\inorm{\psi}$, possibly of the form $\ln(\inorm{\psi})^k$ is inherent even in
the query model, by constructing a distribution $\psi$ (with large
$\inorm{\psi}$) where the query model and random samples model are equivalent,
and where finding skewed subcubes lets us solve the $k$-sparse noisy
parity problem.

\subsection{Related Work}

In nearby work, \cite{AlonAK+07} study the problem of testing whether
a distribution is $\delta$ close in statistical distance to $(\eps,
k)$-wise uniform. In our language, a distribution $D$ is said to be
$(\eps, k)$-wise uniform if all subcubes of codimension $j \leq k$
have skew no greater than $2^j \eps$. They provide a sample complexity
upper bound of $O((k \log n) / \eps^2 \delta^2)$. They then provide
evidence, based on the conjectured hardness of finding planted
cliques, that no polynomial in $n$ time algorithm for this problem
exists (one can also base this hardness on sparse noisy parity, as
we do here). Indeed, their testing algorithm essentially reduces the
problem to the optimization version: find the subcube of codimension
$k$ such that the skew is maximized.

Fourier analytic techniques have found widespread use in a variety of
supervised learning problems under the uniform distribution
\cite{o2014analysis}. Our work differs from this in that the problem
we consider is an unsupervised learning problem, and that we use
Fourier analysis over the uniform distribution to reason about
the deviation from an arbitrary distribution. In this aspect, our work
is similar to the work of \cite{AlonAK+07, ODonnellZ18}.

Finally, there have been a line of recent results in machine learning
which have a list-decoding flavor to them, see for instance
\cite{CharikarSV17, KarmalkarKK19, raghavendra2020list}.

\medskip

\textbf{Outline of the paper.} Section \ref{sec:defs} introduces
definitions and notation. Section \ref{sec:fourier} contains Fourier analytic
results that are required for our results. Section
\ref{sec:comb} proves our main combinatorial bounds on the
number of skewed subcubes, and gives examples that show these bounds
are tight.  Section \ref{sec:algo} describes the efficient algorithm
for enumerating minimal skewed subcubes. Section \ref{sec:noisy-parity} gives
lower bounds due to the reduction from the noisy parity
problem. Section \ref{sec:mem-query} considers the problem in
the membership query model. Some proofs are deferred from the main
body to the Appendix \ref{sec:appendix}.

\section{Definitions}
\label{sec:defs}

In this section we present basic definitions and facts.
\paragraph{Distributions.}
We denote the $n$-dimensional Hamming cube by $\pmo^n$.
Given a
probability distribution $\psi$ on $\pmo^n$, it is convenient to
identify it with the probability measure $\psi: \pmo^n \rgta \R^{\geq 0}$ which satisfies
\[
\E_{\x \sim \pmo^n} [\psi(\x)] = 1.
\]
We write $\x \sim \psi$ to denote $\x$ is a random variable with the
distribution
\[
\Pr_{\x \sim \psi}[\x = x] = \frac{\psi(x)}{2^n}
\]
Henceforth, we will interchangeably refer to $\psi$ as a distribution
and a measure. We will use $\mu$ to denote the uniform distribution over
$\pmo^n$, where $\mu(x) =1$ for all $x \in \pmo^n$.
Given functions $f, g: \pmo^n \rgta \R$, we define their inner product by
$\ip{f,g} \defeq \E_{\x \sim \mu}[f(\x)g(\x)]$. We define $\pnorm{f}{p}
\defeq \E_{\x \sim \mu}[f(\x)^p]^{1/p}$. For two probability measures
$\psi, \theta$, we have
\[ \ip{\psi, \theta} = \sum_{x \in \pmo^n} \frac{\psi(x)\theta(x)}{2^n} =
  \E_{\x \sim \psi}[\theta(\x)] = \E_{\x \sim \theta}[\psi(\x)].\]

Let $A \subseteq \pmo^n$ and let $\alpha = |A|/2^n$ denote its
fractional density. We use $\mu_A$ to denote the uniform distribution
over $A$. The corresponding measure is defined as
\begin{align}
  \label{eq:mu_a}
  \mu_A(\x) =
  \begin{cases}
    \frac{1}{\alpha} & \text{if} \ x \in A\\
    0 & \text{otherwise}
  \end{cases}
\end{align}

For a distribution $\psi$, we define \[ \inorm{\psi} = \max_{x \in \pmo^n} \psi(x)\]
It follows from the definition that
\[ \inorm{\psi} = 2^n\max_{x \in \pmo^n}\Pr_{\x \sim \psi}[\x =
    x] =  \max_{x \in  \pmo^n}\frac{\Pr_{\x \sim \psi}[\x = x]}{\Pr_{\x \sim \mu}[\x = x]} \]
A bound on $\inorm{\psi}$ implies that no point is too likely.

\paragraph{Subcubes.}

A subcube is a subset of $\pmo^n$ obtained by fixing some subset of
bits to a particular value. Formally, a subcube $C \subseteq \pmo^n$
is specified by a pair $(K, y)$ where $K \subseteq [n]$ and
$y \in \pmo^{K}$. We have
\[ C = \{x \in \pmo^n \ \st \  x_i = y_i \ \forall \ i \in K \}. \]
We refer to coordinates in $K$ as the fixed coordinates of $C$, and
to the rest as the free coordinates of $C$. For $C = (K,y)$ we define
the codimension of $C$ to be $|K|$ and denote it $\codim(C)$.  We use
$\C^{\leq k}$ to denote the set of all subcubes of codimension at most $k$.
By Equation \eqref{eq:mu_a}, $\mu_C$ the uniform measure over $C$ is given by
  \begin{align*}
    \mu_C(x)    & = \begin{cases}
      2^{\codim(C)} & \ \text{if} \ x \in C\\
      0 & \ \text{if} \ x \not\in C\\
    \end{cases}
  \end{align*}

For subcubes $C = (K,y), D = (L,z)$, we have $D \subset C$ iff $K
\subset L$ and $z_i = y_i$ for all $i \in K$. We refer to $D$ as a
child of $C$ and $C$ as a parent of $D$.

\begin{definition}(Restriction)
  \label{def:restr}
For a distribution $\psi$ on $\pmo^n$ and a subcube $C \subseteq
\pmo^n$ such that $\psi$ assigns non-zero probability to $C$
we define $\restr{\psi}{C}: C \rightarrow  \R^{\geq 0}$, the
restriction of $\psi$ to $C$, as
\begin{align*}
  \restr{\psi}{C} (x)  =  \frac{\psi(x)}{\ip{\psi, \mu_C}}.
\end{align*}
\end{definition}
Since $\ip{\psi, \mu_C} = \expectarg{\psi(x)}{x \sim \mu_C}$, $\psi$
assings non-zero probaility to $C$ iff $\ip{\psi, \mu_C} > 0$.
The restriction is itself a legal probability measure; it satisfies $\E_{x
  \in C}[\restr{\psi}{C}]  =1$. This definition
immediately implies the relationship:
\begin{fact}
  \label{fact:inorm}
	For a distribution $\psi$ on $\pmo^n$ and a subcube $C \subset \pmo^n$
	\[\inorm{\restr{\psi}{C}} = \frac{\inorm{\psi}}{\ip{\psi, \mu_C}}\]
\end{fact}

\begin{restatable}{lemma}{prodrule}
  \label{lem:prod_rule}
Given subcubes $C$ and $D$ such that $D \subseteq C \subseteq \pmo^n$,
and a density function $\psi$, it holds that:
\[\ip{\psi, \mu_D} = \ip{\psi, \mu_C}\cdot \ip{\restr{\psi}{C}, \mu_D \big|_C}\]
\end{restatable}

\subsection{Skewed subcubes}

\begin{definition}[Skew]
  \label{def:skew}
	We define the \textbf{skew} of a subcube $C$ with respect to
        measure $\psi$ as
	\[ \skew_\psi(C) = \ip{\psi, \mu_C} - 1\]
\end{definition}

The next two lemmas state some simple facts about the skew of
subcubes. First we show the skew of a subcube measures the deviation
of the measure on the subcube from the uniform distribution.

\begin{restatable}{lemma}{skewbyprob}
	\label{lem:skew_by_prob}
	  Let $\codim(C) = k$. We have
	\begin{align*}
	  \skew_\psi(C) &= 2^k \Pr_{\x \sim \psi}[\x \in C]  -1 \\
          &= \frac{1}{\Pr_{x \sim \mu}[\x \in C]}\left( \Pr_{x
	  \sim \psi}[\x \in C] - \Pr_{x \sim \mu}[\x \in C]\right).
	\end{align*}
\end{restatable}

\begin{corollary}
  For any distribution $\psi$,
  $\skew_\psi(C)$ lies in the range $[-1, 2^k -1]$.
\end{corollary}

If $\skew_\psi(C) < 0$, we say that $C$ is negatively skewed while if
$\skew_\psi(C) > 0$ we say that it is positively skewed. An averaging
argument shows that the existence of negatively skewed subcubes
implies the existence of positively skewed subcubes and vice versa.

\begin{restatable}{lemma}{negposskew}
  \label{lem:neg_pos_skew}
  For any $K \subseteq [n]$, we have
  \[ \sum_{\substack{D = (K, w) \\ w \in \pmo^K}} \skew(D) = 0. \]
\end{restatable}

Given a cube $C = (K, y)$ of codimension $k$, we can partition it into
$2^\ell$ subcubes of codimension $k + \ell$, where we pick a set $L$ of $\ell$
additional coordinates outside of $K$ to fix and enumerate over all
settings of these coordinates.

\begin{restatable}{lemma}{skewavg}
	\label{lem:skew_avg}
  If $\{C_1, \ldots, C_{2^\ell}\}$ is a partition of $C$, then
\[ \skew_\psi(C) = \frac{1}{2^\ell}\sum_{i=1}^{2^\ell}\skew_\psi(C_i). \]
\end{restatable}

Having established basic properties of the skew function, we next turn
to bounding the number of subcubes with a given skew. We show that
this number may be quite large in the worst case.

\begin{lemma}
\label{lem:comb_lb}
Let $\gamma = 2^f - 1$ for $f \in \{1, \ldots, k\}$.
There exists a distribution
  $\psi$ such that there are $\Omega((n/k)^k)$ many
  subcubes of codimension $k$ with $\skew(C) \geq \gamma$.
\end{lemma}
\begin{proof}
  Let $C$ be the subcube where the first $t \geq f$ bits are fixed to $1$, and let $\mu_C$ be the uniform distribution over it.
  Consider any subcube $D$ where we choose $f$ indices from $[t]$ and $k - f$ indices from $[n]
  \setminus [t]$, and set them to $1$.
  We have
  \begin{align*}
    \ip{\mu_C, \mu_D} & = \E_{\mu_C}[\mu_D]  = 2^k\Pr_{\x \sim \mu_C}[\x \in D]
    = 2^k\frac{1}{2^{k -f}} = 2^f \geq 1 + \gamma
  \end{align*}
since a point from $\mu_C$ lies in $D$ iff the $k -f$ bits from
$[n]\setminus [t]$ are all set to $1$. We now optimize the choice of
$t$. Let $\alpha = f/k$ for $\alpha \leq 1$. We choose $t = \alpha n$
(ignoring floors and ceilings which will not affect the asymptotics). The
number of choices for $D$ is given by
\[ \binom{t}{f}\cdot \binom{n -t}{k -f}  = \binom{\alpha n}{\alpha
  k}\cdot \binom{(1- \alpha)n}{(1 - \alpha)k} \geq
  \left(\frac{n}{k}\right)^{\alpha k}\cdot \left(\frac{n}{k}\right)^{(1
    - \alpha)k} \geq  \left(\frac{n}{k}\right)^k. \qedhere \]
\end{proof}

While the above bound is proved for positive skew, Lemma
\ref{lem:neg_pos_skew} can be used to derive a similar bound for negative
skew.  Given that this bound is not too far from the trivial upper bound of
$\binom{n}{k}$, we need to refine our notion of skew, and also to
restrict the set of distributions we consider.

\subsection{Minimal skewed subcubes}

Lemma \ref{lem:skew_avg} tells us that if there exists $C = (J, y)$
such that $|J| = j < k$ and $\skew(C) \geq \gamma$, then for any $L \subseteq [n]
\setminus J$ of size $k - j$, there exists some further restriction of
bits in $L$ such that the resulting subcube $D \subseteq C$ has $\skew(D) \geq
\gamma$. This suggests that we ought to ignore subcubes such as $D$ that can
be viewed as {\em inheriting} skew from some parent $C$, and instead
focus on subcubes whose skew is larger than any parent.
One technical issue is that we now need to handle the case of positive and negative
skew separately. This motivates the following definitions.

\begin{definition}
\label{def:min_skew}
  Let $\gamma \in(0,2^k -1]$ and $\eps \in (0,1]$. A subcube $C
    \subseteq \pmo^n$ is a $(\gamma, \eps)$-minimally
  skewed subcube if $\skew(C) \geq \gamma$ and for all its parent
  subcubes $D \supsetneq C$, we have
  \begin{align}
    \label{eq:min_skew}
    \skew_\psi(D) \leq (1 - \eps)\gamma.
  \end{align}
Let $\gamma \in (0,1]$ and $\eps \in (0,1]$. A subcube $C \subseteq
  \pmo^n$ is a $(-\gamma, \eps)$-minimally skewed subcube if
  $\skew(C) \leq -\gamma$ and for all its parent  subcubes $D \supsetneq C$, we have
  \begin{align}
    \label{eq:min_skew_neg}
    \skew_\psi(D) \geq -(1 - \eps)\gamma.
  \end{align}
\end{definition}
Note that our convention is to always use $\gamma >0$ for
the magnitude of the skew, and specify its sign explicitly. Note that
the allowable values of $\gamma$ are different for the case of
positive and negative skew. We restrict $\eps \in (0,1]$. The case
  $\eps = 1$ corresponds to the case where every
  subcube of $C$ has no skew.

The crux of this definition is  that minimal skew cannot be
inherited from a parent. Given a minimal skewed subcube $C$, and a
parent $D \supsetneq C$, we show that $C$ has noticeable skew in the
restriction $\restr{\psi}{D}$.

\begin{lemma}
\label{lem:cond_skew}
If $C$ is a $(\gamma, \eps)$-minimal skewed subcube
and $D \supseteq C$ is a parent of $C$, then
\[ \skew_{\restr{\psi}{D}}(C) \geq \frac{\eps \sqrt{\gamma}}{2}. \]
If $C$ is a $(-\gamma, \eps)$-minimal skewed subcube
and $D \supseteq C$ is a parent of $C$, then
\[ \skew_{\restr{\psi}{D}}(C) \leq -\eps\gamma. \]
\end{lemma}
\begin{proof}
  We first consider the case when $\gamma > 0$.
  By \cref{lem:prod_rule}
  \begin{align}
    \label{eq:lb_ip}
	\ip{\restr{\psi}{D}, \restr{\mu_C}{D}}  & =
	\frac{\ip{\psi, \mu_C}}{\ip{\psi, \mu_D}}
	= \frac{1 + \skew_\psi(C)}{1 + \skew_\psi(D)}
	\geq \frac{1 + \gamma}{ 1 + (1 - \eps)\gamma}.
	\end{align}
	We have
	\begin{align*}
	\skew_{\restr{\psi}{D}}(C) &= \ip{\restr{\psi}{D},
		\restr{\mu_C}{D}} - 1 \geq
	\frac{\eps\gamma}{1 + (1 - \eps)\gamma} \geq
        \frac{\eps\gamma}{2\sqrt{(1 - \eps)\gamma}} \geq \frac{\eps \sqrt{\gamma}}{2}
	\end{align*}
	where the first inequality is by Equation \eqref{eq:lb_ip} and
        the second is by the AM-GM inequality.

        Next we consider the case where $\gamma < 0$.
	By \cref{lem:prod_rule}
	\begin{align*}
	\ip{\restr{\psi}{D}, \restr{\mu_C}{D}}  & =
	\frac{\ip{\psi, \mu_C}}{\ip{\psi, \mu_D}}
	= \frac{1 + \skew_\psi(C)}{1 + \skew_\psi(D)}
	\leq \frac{1 - \gamma}{ 1 - (1 - \eps)\gamma}.
	\end{align*}
        Hence
        \begin{align*}
          \skew_{\restr{\psi}{D}}(C) &= \ip{\restr{\psi}{D}, \restr{\mu_C}{D}}  -1
          \leq \frac{1 - \gamma}{ 1 - (1 - \eps)\gamma} - 1
          = \frac{\eps\gamma}{1 - (1 - \eps)\gamma} \leq - \eps\gamma. \qedhere
       \end{align*}
\end{proof}

\section{Fourier Analysis}
\label{sec:fourier}

Given $S \subseteq [n]$, let $\chi_S: \pmo^n \rgta \pmo$ be given by
$\chi_S(x) = \prod_{i \in S}x_i$. These functions form a basis so we
can write $\psi = \sum_S \widehat \psi(S) \chi_S$, where the  Fourier
coefficients of $\psi$ are given by
\begin{align*}
  \hatpsi(S) = \E_{\x \sim \mu}[\psi(\x)\chi_S(\x)] = \sum_{x \in \pmo^n}\frac{\psi(x)\chi_S(x)}{2^n}
  = \sum_{x \in \pmo^n}\Pr_{\x \sim \psi}[\x = x]\chi_S(x) = \E_{\x \sim \psi}[\chi_S(\x)]
\end{align*}
which is simply the bias of $\chi_S$ under the distribution
$\psi$. Thus we have
\[ \psi(x)  = \sum_{S \subseteq  [n]}\hatpsi(S)\chi_S(x) \]
where $\hatpsi(\emptyset) = \E_{\x \sim \psi}[\psi(x)] = 1$.
Given two distributions $\psi$ and $\omega$, their inner product is
given by
\begin{align*}
  \ip{\psi, \omega} = \E_{\x \sim \pmo^n}[\psi(\x)\omega(\x)]
  = \sum_{x \in \pmo^n}\frac{\psi(x)\omega(x)}{2^n}
  = 1 + \sum_{\emptyset \neq S \subseteq [n]}\hatpsi(S)\hatomega(S)
\end{align*}

\subsubsection*{Skew implies heavy low-degree coefficients}

We show that large skew in the subcube $(K, y)$
implies non-trivial Fourier mass on subsets of $K$.

\begin{lemma}
\label{lem:skew_fourier}
  For $C = (K,y)$,
  \begin{align*}
  \skew_\psi(C)	&= \sum_{\emptyset \neq S \subseteq K} \hatpsi(S)\chi_S(y).
  \end{align*}
\end{lemma}
\begin{proof}
  Given $C =(K, y)$, $\mu_C$ the uniform measure over $C$ is given by
  \[ \mu_C(x) = \prod_{i \in K}(1 + x_iy_i) = 1 + \sum_{\emptyset \neq S
    \subseteq K}\chi_S(y)\chi_S(x). \]
Hence we have
  \begin{align*}
\ip{\psi, \mu_C} = 1 + \sum_{\emptyset \neq S \subseteq
	[n]}\hatpsi(S)\hatmu_A(S) =  1 + \sum_{\emptyset \neq S
  \subseteq K}\hatpsi(S)\chi_S(y)
  \end{align*}
  from which the claim follows.
\end{proof}

Given the above lemma, our approach is to reduce bounding the number of
skewed subcubes to bounding the number of large Fourier coefficients
of  $\psi$ at level $k$.

We define
\[ W^{\leq k}(\psi) = \sum_{\substack{S \subseteq [n] \\ |S| \leq k}} \hatpsi(S)^2.\]
A trivial bound is
obtained from  Parseval's identity:
\begin{align*}
 W^{\leq k}(\psi) \leq  \sum_{S \subseteq [n]} \hatpsi(S)^2   &=  \E_{\x \sim \pmo^n}[\psi(\x)^2]
 \leq \inorm{\psi}\E_{\x \sim \pmo^n}[\psi(\x)] = \inorm{\psi}
\end{align*}
If we restrict the summation to sets $S$ of cardinality at most $k$,
then a much stronger bound of $O(\ln(\inorm{\psi})^k)$ holds, it is proved
using the powerful HyperContractivity Theorem. These bounds
generalize the Level-k inequalities for the Fourier spectrum of
small-sets, indeed the proof is identical.

For a $f: \pmo^n \rgta \R$ and $0 \leq \rho \leq 1$ set
\begin{align*}
T_\rho f = \sum_S \rho^{|S|}\widehat f(S)\chi_S
\end{align*}

$T_\rho$ is known as the noise operator. Recall that for $p > 0$ we have $\norm{f}_p = \E[f^p]^{1/p}$. The hypercontractive inequality quantifies the extent to which the noise operater reduces the norm of a function. See for instance.
\cite[Chapter 2]{o2014analysis} for a detailed exposition.

\begin{theorem}
\label{thm:hc}
  Let $f: \pmo^n \rgta \R$ and $\rho \in [0,1]$. Then
  \[ \norm{T_\rho f}_2 \leq \norm{f}_{1 + \rho^2}. \]
\end{theorem}

We use the hypercontractive inequality to bound the mass of the low level coefficients.

\begin{theorem}
  \label{thm:level_k}
  Let $\psi$ be a distribution. Then
  \[ W^{\leq k}(\psi) \leq e^2 (\ln(e \inorm{\psi}))^k. \]
\end{theorem}
\begin{proof}
  By \cref{thm:hc}, we have
  \begin{align*}
    \norm{T_\rho \psi }_2 & \leq \norm{\psi}_{1 + \rho^2}\\
    & = \left(\E_{\x \sim \pmo^n}[\psi(x)^{\rho^2} \psi(x)]\right)^{1/(1 +
      \rho^2)}\\
    & \leq \inorm{\psi}^{\rho^2/(1 + \rho^2)}\E_{\x \sim \pmo^n}[
      \psi(x)]^{1/(1 + \rho^2)}\\
    & =  \exp\left(\frac{\ln(\inorm{\psi})\rho^2}{1 + \rho^2}\right)
  \end{align*}
  where we used Holder's inequality with $p = \infty$, $q =1$.
  Taking $\rho = \min(1, 1/\sqrt{\ln(\inorm{\psi})})$, we have
   \[\norm{T_\rho \psi }_2  \leq \exp(1/(1 + \rho^2)) \leq e\]
   But note that
   \[ \norm{T_\rho \psi }_2  \geq (\rho^{2k} W^{\leq k}(\psi))^{1/2}\]
   Hence we conclude that
   \[W^{\leq k}(\psi)) \leq e^2(1/\rho)^{2k} = e^2 \max(1, \ln \inorm{\psi})^k \leq e^2(\ln(e \inorm{\psi}))^k. \qedhere\]
\end{proof}

We also need a bound for the Fourier mass at level $k$ where we do not
count coordinates from some set $J \subseteq [n]$ in the degree of a
coefficient.
\[ W^{\leq k}(\psi, J) = \sum_{T \subseteq J} \sum_{\substack{S
      \subseteq [n]\setminus J \\ |S| \leq k}}\hatpsi(S \cup T)^2.\]

\begin{restatable}{corollary}{corlevel}
  \label{cor:level_k}
  For $J \subseteq [n]$ and a distribution $\psi$ over $\pmo^n$,
  \[ W^{\leq k}(\psi, J)   \leq 2^{|J|}e^2 (\ln(e \inorm{\psi}))^k. \]
\end{restatable}

\subsubsection*{Projection, Extension, Restriction}

Given $x \in \pmo^n$ and a set of coordinates $P \subseteq [n]$, let
$x_{P}$ denote the projection of $x$ onto coordinates in $P$.
Given a distribution $\psi$ over $\pmo^n$ and a set of coordinates $P
\subseteq [n]$, let $\psi_{P}$ denote the marginal distribution over
the set $P$. The Fourier expansion is especially convenient for
marginals, we simply restrict the sum to subsets of $P$.

\begin{restatable}{lemma}{marginal}
  \label{lem:marginal}
For $P \subseteq [n]$ and a distribution $\psi$ over $\pmo^n$, the
restriction $\psi_{P}$ is given by
\[ \psi_{P}(y)  = \sum_{S \subseteq  P}\hatpsi(S)\chi_S(y). \]
\end{restatable}

Coversely we can {\em extend} a distribution $\psi'$ defined on
$\pmo^P$ for  $P \subseteq [n]$ to all of $\pmo^n$ while preserving
its important properties.

\begin{restatable}{lemma}{extension}
\label{lem:extension}
  Let $P \subsetneq  [n]$. Let $\psi'$ be a distribution on
$\pmo^P$. Define a distribution $\psi$ on $\pmo^n$ by $\psi(x) =
\psi'(x_P)$. Then
\begin{enumerate}
  \item $\psi$ is the product distribution of $\psi'$ with the uniform
    distribution on $\pmo^{\bar{P}}$.
  \item $\inorm{\psi} = \inorm{\psi'}$.
  \item $C$ is a minimal skewed subcube under $\psi$
    iff it is a minimal skewed subcube under
    $\psi'$.
\end{enumerate}
\end{restatable}

Finally, we derive an expression for the Fourier expansion of
$\restr{\psi}{C}$ in terms of the coefficients of $\psi$.

\begin{restatable}{lemma}{restriction}
  \label{lem:restrict}
  Let $C = (J, z)$. Then
  \[ \restr{\psi}{C}(x) = \sum_{S \subseteq [n]\setminus J} \chi_S(x) \frac{\sum_{T
    \subseteq J} \hatpsi(S \cup T)\chi_T(z)}{\ip{\psi, \mu_C}}. \]
\end{restatable}

\section{A Combinatorial Bound for minimal skewed subcubes}
\label{sec:comb}

In this section, we show bounds on the number of minimal skewed
subcubes that is dimension independent.

\begin{theorem}[Combinatorial Bound for Positive Skew]
  \label{thm:comb_pos}
  For any measure $\psi$ on $\pmo^n$, integer $k \leq n$, and $\gamma
  \in (0, 2^k -1]$ and $\eps \in (0,1]$, the number of
  $(\gamma, \eps)$-minimal skewed subcubes of codimension at most $k$
      is bounded by
      \[k^{O(k)} \left(\frac{1}{\eps^2\gamma}\ln(e \inorm{\psi})\right)^k.\]
\end{theorem}

\begin{theorem}[Combinatorial Bound for Negative Skew]
  \label{thm:comb_neg}
  For any measure $\psi$ on $\pmo^n$, integer $k \leq n$, and $
  \gamma \in (0,1]$ and $\eps \in(0, 1]$, the number of
  $(-\gamma, \eps)$-minimal skewed subcubes of codimension at most $k$
  is bounded by
        \[k^{O(k)} \left(\frac{1}{\eps^2\gamma^2}\ln\left(\frac{e
          \inorm{\psi}}{\eps\gamma}\right) \right)^k.\]
\end{theorem}

We now outline our approach for proving these bounds.
\begin{enumerate}

\item   We give an algorithm to enumerate all minimal skewed
  subcubes, given the list of large, low-degree
  Fourier coefficients in an adaptively chosen sequence of restrictions of the original
  distribution $\psi$. The algorithm
  recursively `grows' skewed subcubes by finding heavy Fourier coefficients and
  restricting the bits in that coefficient, and showing that this
  algorithm discovers all minimal skewed subcubes.

\item We bound the number of large low-degree Fourier
  coefficients of $\psi$ using Theorem \ref{thm:level_k}.

\end{enumerate}

The details of the algorithm in Step 1 are different for the cases of
positive and negative skew, so we present them in Subsections
\ref{subsec:pos} and \ref{subsec:neg}. To go from a combinatorial
bound to an efficient algorithm, we need to make Step 2
algorithmic. We will consider this problem under different learning
models in Sections \ref{sec:algo} and \ref{sec:mem-query}.

\subsection{Positive skew}
\label{subsec:pos}

We first present an algorithm $\fsr$ for enumerating minimal skewed
subcubes where the skew is positive.

To prove the combinatorial bound, we allow the algorithm to make
certain {\em guesses} in \cref{line:guess_s,line:guess_z}.
We think of the set of all possible outputs over all possible guesses as the list
that is returned by the algorithm. In Lemma \ref{lem:complete}, we will show that all minimal
skewed subcubes are contained in this list.  We bound the list size in
Lemma \ref{lem:comb_bound}. Together, these complete the proof of
Theorem \ref{thm:comb_pos}.

        We start the recursion with $R_0 = \emptyset$ and $z_0 = \emptyset$ the
        null string. The routine either returns FAIL or returns
        $S_t \subset [n]$ and $z_t \in \pmo^{S_t}$ such that $(R_t,
        z_t)$ is a $\gamma$-skewed subcube. The
        algorithm also takes as inputs the input the distribution $\psi$,
        a bound $k$ on the codimension, and skew parameters $\gamma
        \in (0,2^k -1]$ and  $\eps \in (0,1]$.
        These stay constant through the recursion, so
        we suppress the dependence on them.
        Consider the list of all possible choices returned by the
        algorithm.

	\begin{algorithm}[ht]
		\caption{$\fsr(R_{t}, y_{t})$}
		\label{alg:find_skew_recursive}
		\begin{algorithmic}[1]
		  \State Let $D_{t} = (R_{t}, y_{t})$.  Let $\psi_t =
                  \restr{\psi}{D_{t}}$. Let $k_{t} = k - |R_t|$.
		\If{$\skew_\psi(D_{t})  > \gamma(1 - \eps)$}\\
			\Return $D_t$ \label{line:return}
		\EndIf
		\If{$\ip{\psi, \mu_{D_t}} < (1 + \gamma)\cdot 2^{-k_t}$}\label{line:test}\\
			\Return FAIL \label{line:fail}
			\EndIf
                \State Pick $S_t$ such that $|S_t| \leq k_{t}$ and \label{line:guess_s}
                  \begin{align}
                    \label{eq:guess_s}
                    |\widehat{\psi_t}(S_t)| \geq
                        \frac{\eps\sqrt{\gamma}}{ k_{t} \cdot
                          \binom{k_{t}}{|S_t|}}.
                  \end{align}
                  \State Pick $z_t \in \pmo^{S_t}$. \label{line:guess_z}
                  \State \Return $\fsr(R_t \cup S_t, y_t \circ z_t)$.
		\end{algorithmic}
	\end{algorithm}

We need some notation for the analysis. Let the sequence of subcubes
produced by the algorithm be $D_0 \supsetneq D_1 \cdots \supsetneq D_\ell$.
Let $s_t = |S_t|$.

        \begin{lemma}
          \label{lem:complete}
          For every $(\gamma,\eps)$-minimal skewed subcube $C$ with $\codim(C) \leq  k$
          there are sequences of choices of $S_t$ and  $z_t$ (in
          \cref{line:guess_s,line:guess_z}) so that $C$ is returned by
          $\fsr$.
        \end{lemma}
        \begin{proof}
        For every $C = (K^*, z^*)$ that is a $(\gamma, \eps)$-minimal
        skewed subcube where $\codim(C) \leq k$, we will show that for every $t$, if
        $D_{t} \supsetneq C$ is parent of $C$, and is not equal to $C$
        there is a choice of $S_t, z_t$ that leads to a parent $D_{t+1}$ of $C$ with a larger
        codimension. Since $t \leq \codim(D_t) \leq \codim(C) \leq k$,
        in $\ell \leq k$ steps
        we must have $D_t = C$, at which point we return at
        \cref{line:return}. Thus the claim implies the lemma.

        At $t = 0$, we have $D_0 =\pmo^n$ so the parent
        condition holds trivially. Assume that we have $D_t \supsetneq C$.

        By the definition of a minimal skewed subcube, $\skew(D_{t})
        \leq (1- \eps)\gamma$, hence the procedure will not return at
        \cref{line:return}.

        Next we show that $\ip{\psi, \mu_{D_t}} \geq (1 +
          \gamma)2^{-k_t}$, the algorithm
          will not return FAIL at \cref{line:fail}:
        \begin{align*}
          \ip{\psi, \mu_{D_t}} & = \Pr_{\x \sim \psi}[\mu_{D_t}(\x)] =
          2^{k - k_t} \Pr_{\x \sim \psi}[\x \in D_t] \\
          & \geq 2^{k - k_t}\Pr_{\x \sim \psi}[\x \in C]\\
          & \geq 2^{-k_t}\ip{\psi, \mu_{C_t}}\\
          & \geq (1 + \gamma)2^{-k_t}.
          \end{align*}
        The first inequality holds because $D_t \supset C$, the
        second because $\codim(C) \leq k$ and
        the last because we assume that $\skew(C) \geq \gamma$.

        Recall that $\psi_t = \restr{\psi}{D_{t}}$, and let $K_{t} = K
        \setminus S_{t}$. By \cref{lem:skew_fourier},
        \begin{align*}
		\skew_{\psi_t}(C) =  \sum_{\emptyset \neq S
                    \subseteq K_{t}}\hatpsi_t(S)\chi_S(y) = \sum_{k'}
                  \sum_{\substack{\emptyset \neq S \subseteq K_{t} \\ |S|
                      = k'}}\hatpsi_t(S)\chi_S(y)
	\end{align*}
        By \cref{lem:cond_skew},
        $\skew_{\psi_t}(C) \geq \eps\sqrt{\gamma}/2$
	which implies that for some $k' \leq k_{t}$, we have
		\[\sum_{\substack{\emptyset \neq S \subseteq K_{t}
                    \\ |S| = k'}}\hatpsi_t(S)\chi_S(y) \geq
                \eps\sqrt{\gamma}/k_{t}\]
		which in turn implies that for at least one $\emptyset
                \neq S_t \subseteq K_{t}$, we have
                \[ |\hatpsi_t(S_t)| \geq
                \eps\sqrt{\gamma} / \left(k_{t}\cdot \binom{k_{t}}{k'} \right).\]
        Assume that we pick this $S_t$ in \cref{line:guess_s} and $z_t
        = \restr{z^*}{S}$ in \cref{line:guess_z}. This ensures
        that $D_{t+1}$ is a parent of $C$ of larger codimension.
        \end{proof}

We next bound the number of all possible outputs of the algorithm. The
crux of the argument is to bound the number of large low-degree
Fourier coefficients using \cref{thm:level_k}. This in turn requires a
bound on the infinity norm of $\psi_t$ which comes from passing the
test in Line \ref{line:test}.

          \begin{lemma}\label{lem:guess_s}
            The number of choices for $S_t$ satisfying Equation
            \eqref{eq:guess_s} is bounded by
            \[ \frac{e^2}{\eps^2\gamma} (\ln(2^{k_t} e \inorm{\psi}))^{s_t}k_t^{4s_t + 2}.\]
          \end{lemma}
          \begin{proof}
          We bound  $\inorm{\psi_t}$ as
          \[\inorm{\psi_t} = \inorm{\restr{\psi}{D_t}} \leq \frac{\inorm{\psi}}{\ip{\psi,
              \mu_{D_t}}} \leq  \frac{\inorm{\psi}}{(1+ \gamma)2^{-k_t}} \leq \inorm{\psi}2^{k_t}\]
           where the first inequality is from \cref{fact:inorm} and we
           have $\ip{\psi, \mu_{D_t}} \geq (1+ \gamma)2^{-k_t}$ since
           we check for this condition in \cref{line:fail}. We now
           use \cref{thm:level_k} which gives
           \[ W^{\leq s_t}(\psi) \leq e^2 (\ln(e 	\inorm{\psi}\cdot
             2^{k_t}))^{s_t}. \]
           Hence the number of choices for $S_t$  satisfying \eqref{eq:guess_s} is bounded by
           \begin{align*}
             W^{\leq s_t}(\psi_t)\left(\frac{k_t\binom{k_t}{s_t}}{\eps\sqrt{\gamma}}\right)^2
             &\leq \frac{e^2}{\eps^2\gamma} (\ln(2^{k_t} e \inorm{\psi}))^{s_t}k_t^{4s_t + 2}. \qedhere
          \end{align*}
          \end{proof}

  \begin{lemma}
		\label{lem:comb_bound}
		The total number of subcubes of codimension $k$ output
                by $\fsr$ is at most:
		\[k^{O(k)} \left(\frac{\ln(e \inorm{\psi})}{\eps^2\gamma}\right)^k.\]
	\end{lemma}

	\begin{proof}
Since $\sum_{t \leq \ell} s_t = k$, the sequence $\{s_t\}_{t=1}^\ell$ is a partition of
$k$, and there are at most $k^k$ of them. Let us fix the sequence.
The number of choices for $S_t$ is bounded by
Lemma \ref{lem:guess_s}.  Since $z_t \in \pmo^{S_t}$, the number of
choices for $z$ is $2^{s_t}$.  Taking the product over all $t$,
the number of possible outputs for $\fsr$ is bounded by
\[  \prod_{t=1}^\ell \frac{e^2}{\eps^2 \gamma} (\ln(2^{k_t}e \inorm{\psi}))^{s_t}k_t^{4s_t + 2}\cdot 2^{s_t} \]

                  We can bound
                  \[ \prod_{t=1}^\ell (\ln(2^{k_t} e \inorm{\psi}))^{s_t}
                  \leq  \ln(2^k  e  \inorm{\psi})^k \leq  (k  +  \ln(e
                  \inorm{\psi}))^k \leq (2k)^k (\ln(e \inorm{\psi}))^k.\]
                  \[\prod_{t=1}^\ell k_t^{4s_t +2}2^{s_t} \leq k^{5k +2}.\]
                  Including the $k^k$ choices  for $s_1, \ldots, s_t$,
                  the output list size is bounded by
                  \[ \left(\frac{e^2}{\eps^2\gamma}\right)^k (\ln(e
                  \inorm{\psi}))^k    k^{7k    +   2}    =    k^{O(k)}
                  \left(\frac{\ln(e
                    \inorm{\psi})}{\eps^2\gamma}\right)^k. \qedhere \]
        \end{proof}

Together Lemma \ref{lem:complete} and Lemma \ref{lem:comb_bound}
complete the proof of Theorem \ref{thm:comb_pos}.

\subsection{Negative Skew}
\label{subsec:neg}

We now present an algorithm \fsn for the negative
skewed case. The algorithm takes as input $\gamma \in (0,1]$ and $[\eps \in
(0,1]$ and the goal is to list all $(-\gamma, \eps)$-minimal
negatively skewed subcubes.

	\begin{algorithm}[ht]
		\caption{$\fsn(R_{t}, y_{t})$}
		\label{alg:find_skew_neg}
		\begin{algorithmic}[1]
		  \State Let $D_{t} = (S_{t}, z_{t})$.  Let $\psi_t =
                  \restr{\psi}{D_{t}}$. Let $k_{t} = k - |S_{t}|$.
		\If{$\skew_\psi(D_{t})  < -\gamma(1 - \eps)$}\\
			\Return $D_t$ \label{line:return_neg}
		\EndIf
                \State Pick $S_t$ such that $|S_t| \leq k_{t}$ and \label{line:guess_s_neg}
                \begin{align}
                  \label{eq:large_neg}
                    |\widehat{\psi_t}(S_t)| \geq
                        \frac{\eps\gamma}{ k_{t} \cdot \binom{k_{t}}{|S_t|}}.
                  \end{align}
                  \State Pick $z_t \in \pmo^{S_t}$. \label{line:guess_z_neg}
                  \State \Return $\fsn(R_t \cup S_t, y_t \circ z_t)$.
		\end{algorithmic}
	\end{algorithm}

        The main differences from \fsr are that once the skew is less
        than $-\gamma(1 - \eps)$, we can return. Thus we can combine
        the Return statement (Line \ref{line:return}, and the
        the check in Line \ref{line:fail}. Also, the bound on the
        coefficient size in Equation\eqref{eq:large_neg} now reflects
        the bound for the negative skew case in Lemma \ref{lem:cond_skew}.

We have the following claim about the
correctness of \fsn.

        \begin{lemma}
          \label{lem:neg_complete}
          For every $(\gamma,\eps)$-minimal skewed subcube $C$ with $\codim(C) \leq  k$
          there are choices of subsets $S_t$ and $z_t$ (in
          \cref{line:guess_s,line:guess_z}) so that $C$ is returned by
          $\fsn$.
        \end{lemma}

        We prove this by showing that for every $t$, if
        $D_{t} \supsetneq C$ is parent of $C$, and is not equal to $C$
        there is a choice of $S_t, z_t$ that gives a parent $D_{t+1}$ of $C$ with a larger
        codimension. Indeed, we know that for any parent of $C$,
        inner product $\ip{\psi, \mu_{D_t}}$ is large enough to pass
        the test in Line \ref{line:return_neg}. The rest of the proof
        is identical to that of Lemma \ref{lem:complete} for
        the case of positive skew, so we do not repeat it.

		The crux of the proof is to bound the number of choices for $S_t$ satisfying Equation
		\eqref{eq:large_neg}.
        \begin{lemma}\label{lem:guess_s_neg}
			The number of choices for $S_t$ satisfying Equation
			\eqref{eq:large_neg} is bounded by
			\[ \frac{e^2}{\eps^2\gamma^2} \left(\ln\left(\frac{e
				\inorm{\psi}}{\eps \gamma}\right)\right)^{s_t}k^{2s_t +2}.\]
		\end{lemma}

		\begin{proof}
			To pass Line
			\ref{line:return_neg}, it must hold that
			$\skew_\psi(D_t) \geq -\gamma(1 - \eps)$, hence
			\[ \ip{\psi,
				\mu_{D_t}} \geq 1 - \gamma(1 - \eps) \geq \eps\gamma \]
			since $\gamma \leq 1$.
			So we bound  $\inorm{\psi_t}$ as
			\[\inorm{\psi_t} = \inorm{\restr{\psi}{D_t}} = \frac{\inorm{\psi}}{\ip{\psi,
					\mu_{D_t}}} \leq \frac{\inorm{\psi}}{\gamma\eps}.\]
			Using \cref{thm:level_k} gives
			\[ W^{\leq s_t}(\psi_t) \leq e^2\left(\ln\left(\frac{e
				\inorm{\psi}}{\eps \gamma}\right)\right)^{s_t}. \]
			Hence the number of choices for $S_t$ satisfying Equation
			\eqref{eq:large_neg} is bounded by
			\begin{align*}
			\frac{W^{\leq s_t} \cdot  \left(k_{t} \cdot
				\binom{k_{t}}{|S_t|}\right)^2}{\eps^2\gamma^2}
			& \leq  e^2\left(\ln\left(\frac{e \inorm{\psi}}{\eps
				\gamma}\right)\right)^{s_t}\frac{\left(k_{t} \cdot
				\binom{k_{t}}{|S_t|}\right)^2}{\eps^2\gamma^2} \\
			& \leq \frac{e^2}{\eps^2\gamma^2} \left(\ln\left(\frac{e
				\inorm{\psi}}{\eps \gamma}\right)\right)^{s_t}k^{2s_t +2}. \qedhere
			\end{align*}
			\end{proof}

		We can now conclude as before.
          \begin{lemma}
		\label{lem:comb_bound_neg}
	The total number of subcubes of codimension $k$ output by
        $\fsn$ is bounded by
\[k^{O(k)}\left(\frac{1}{\eps^2\gamma^2}\ln\left(\frac{e
                \inorm{\psi}}{\eps \gamma}\right)\right)^k.\]
          \end{lemma}
          \begin{proof}
          	Using \cref{lem:guess_s_neg} and the fact that there are are $2^{s_t}$ choices for $z_t$ and $k^k$ choices of the partition $s_1, \ldots, s_t$, the overall
           list size is bounded by
           \[ k^k\prod_{t = 1}^\ell\frac{e^2}{\eps^2\gamma^2} \left(\ln\left(\frac{e
                \inorm{\psi}}{\eps \gamma}\right)\right)^{s_t}k^{2s_t
             +2}2^{s_t} \leq k^{O(k)}\left(\frac{1}{\eps^2\gamma^2}\ln\left(\frac{e
                \inorm{\psi}}{\eps \gamma}\right)\right)^k. \qedhere\]
          \end{proof}

        Combining Lemmas \ref{lem:neg_complete} and
        \ref{lem:comb_bound_neg} completes the proof of Theorem
        \ref{thm:comb_neg}.

        \subsection{Tightness of our bounds}
        We show that the dependence on $\inorm{\psi}$ in Theorem
        \ref{thm:comb_pos} is nearly optimal. To simplify our
        constructions, we will construct distributions on $n$ variables where
        $n = n(\inorm{\psi}, k)$. But one can then use Lemma
        \ref{lem:extension} to extend the construction to all larger
        values of $n$.

        \begin{theorem}
          There exists a distribution $\mu_C$ on $\pmo^n$ which has
  $\Omega_k((\ln(\inorm{\mu_C}))^k)$ many $(2^k,1/2)$-minimal
            skewed subcubes of codimension $k$.
\end{theorem}
\begin{proof}
Let $C$ be the subcube where all the
bits are fixed to $1$, and let $\mu_C$ be the uniform distribution
over it. It follows that $\inorm{\mu_C} = 2^n$ hence
$\ln(\inorm{\mu_C}) = n$. We claim that all $\binom{n}{k}$ subcubes
where a $k$ out of the first $t$ bits are fixed to $1$ are $(2^k -1,
1/2)$-minimal skewed subcubes. Fix one such cube $D$. We have
\[ \skew_{\mu_C}(D) = 2^k\Pr_{\x \sim \mu_C}[\x \in D] - 1 = 2^k -1.\]
Since the maximum skew of any subcube of codimension $k-1$ is at most
$2^{k-1} -1$, $D$ satisfies the definition of  $(\gamma,\eps)$-minimal
skew for $\gamma = 2^k -1$ and $\eps$ such that $\gamma(1 - \eps) \geq 2^{k-1} - 1$.
In particular, we can take $\eps = 1/2$.

Thus the number of $(2^k -1, 1/2)$-minimal skewed subcubes is $\binom{n}{k} = \Omega_k(n^k)$.
The only dependence on $n$ in Theorem \ref{thm:comb_pos} comes from
the $\inorm{\mu_C}$ since $\eps$ is a constant and $\gamma
\leq 2^k$. Thus the number of cubes is $\Omega((\ln(\inorm{\mu_C}))^k)$.
\end{proof}

For Theorem \ref{thm:comb_neg} dealing with negative skew, we show a
similar bound, though with a smaller value of $\eps = 1/k$. The
distribution we use is derived from the Tribes function.

 \begin{theorem}
          There exists a distribution $\tau$ on $\pmo^n$ which has
  $\Omega_k((\ln(\inorm{\tau}))^k)$ many $(-1, 1/k)$-minimal
            skewed subcubes of codimension $k$.
\end{theorem}
\begin{proof}
  Let $n  = tk$. We label the coordinates as $\{x_{i,j}\}_{i
    \in [k], j\in [t]}$.
  Consider the DNF formula
  \[ \Tr(x) = \bigvee_{i=1}^{k}\bigwedge_{j=1}^t x_{i,j}\]
We now define a distribution $\tau$ where we pick a $i^* \in [k]$
at random, set $x_{i^*, j} = 1$ for all $j \in [t]$ and set all the
other variables randomly. Clearly the distribution $\tau$ is supported
on the satisfying assignments of $\Tr(x)$. It is also easy to see that
\[ \inorm{\tau} = \tau(1^{tk}) = 2^{tk}\sum_{i=1}^k\frac{1}{k}2^{-(k
  -1)t} = 2^t. \]

Now consider the set of minimal $0$ certificates of $\Tr$, which are
subcubes where we pick a single variables from each term and set it to
$0$. There are $t^k$ such subcubes, fix one such subcube $C$. Clearly
$\Pr_{\x \sim \tau}[\x \in C] = 0$, hence $\skew_\tau(C) = -1$. Now
consider any parent subcube $D$ of $C$. Assume it has codimension
$\ell < k$, and let $L \subset [k]$ denote the set of terms that it
sets to $0$. For $\x \sim \tau$ to lie in $L$, two events need to
happen:
\begin{itemize}
\item $i^* \not\in L$, which happens with probability $1 - \ell/k$.
\item The variables in $L$ which are set to $0$ in $D$ are also set to
  $0$ by the random assignment, which happens with probability
  $2^{-\ell}$.
\end{itemize}
As these two events are independent, we have $\Pr_{\x \sim \tau}[\x \in
  D] = 2^{-\ell}(1 - \ell/k)$ hence
\[ \skew_\tau(D) = 2^\ell\Pr_{\x \sim \tau}[\x \in
  D] -1 = - \frac{\ell}{k}. \]
Thus the maximum skew of any parent $D$ is $-1 + 1/k$. Hence $C$ is
$(-1, 1/k)$-minimally skewed.

As before we note that $\gamma =1$ and $\eps = 1/k$, hence the only
dependence on $t$ comes from $\log(\inorm{\tau}) = t$, which gives the
claimed bound.
\end{proof}

Finally, we present a distribution that has a large number of $(-1,1)$
and $(1,1)$ minimally skewed subcubes. Recall that $\eps =1$ means that
every parent of the cube has skew $0$.

The construction is based on (dual) BCH codes.
We think of linear codes as subsets of $\F_2^n$ where $\F_2 = \{0,1\}$
which we can identify with $\pmo^n$ via the usual mapping $x \rightarrow (-1)^x$.
For $x \in \F_2^n$ let the weight of $x$ denoted $\wt(x)$ be the number of $1$s in $x$. Let
$\supp(x) \subseteq [n]$ denote the set of cordinates where $x$ is
non-zero. We will use the following fact about BCH codes communicated
to us by Sergey Yekhanin \cite{Yekhanin}.

\begin{restatable}{lemma}{lemsergey}
\label{lem:sergey}
  \cite{Yekhanin}
Let $d \geq 2$ be even and let $n + 1 = 2^l \geq d$. There exists a
$\F_2$-linear code $\C_{BCH} \subseteq  \zo^n$ with minimum distance
$d$, which contains $\Omega(n^{d/2 + 1})$ minimum weight codewords.
\end{restatable}

\begin{theorem}
  For any even $k \geq 2$ and large enough $n$, there exists a distribution $\psi_k$ on $\pmo^n$ where the numbers of
  $(-1,1)$-minimal skewed subcubes and $(1,1)$-minimal skewed subcubes
  of codimension $k$ are both $\Omega_k((\log(\inorm{\psi_k}))^{k/2
    + 1})$.
\end{theorem}
\begin{proof}
  Set $k = d$, and take $n$ as in Lemma \ref{lem:sergey}. Let $\psi^k$ be the
  uniform distribution on the dual space to $\C_{BCH}$.
  Using standard facts
  about the Fourier expansion of a subspace, we can write
  \begin{align}
  \label{eq:fourier_bch}
    \psi^k = \sum_{\substack{c \in \C_{BCH}\\ S = \supp(c)}}\chi_S(x).
  \end{align}

  Since $\C_{BCH}$ has minimum distance $k$, $\psi^k$ is $(k-1)$-wise
  independent, so for any subcube $C$ where $\codim(C) \leq k -1$, we
  have $\skew_{\psi^k}(C) = 0$. This relies on a standard construction
  of $k$-wise independent spaces from codes \cite[Chapter 16]{AlonS},
  it can also be seen using Lemma \ref{lem:skew_fourier} combined with
  Equation \eqref{eq:fourier_bch}.

  Fix  $S \subset [n]$ to be the support of a minimum weight codeword
  in $\C_{BCH}$. By Lemma \ref{lem:marginal} and Equation \eqref{eq:fourier_bch},
  the projection of $\psi^k$ to coordinates in $S$ is
  given by $\psi^k_S(x) = 1 + \chi_S(x)$. Hence it is uniform over the
  $2^{k-1}$ settings $y \in \pmo^S$ such that  $\chi_S(y^+) = 1$.

For every such $y$ and $D^+ = (S, y)$, we have $\Pr_{\x \sim \psi^k}[\x
  \in D^+]   = 2^{-(k-1)}$  hence
    \[\skew_{\psi^k}(D^+) = 2^k \Pr_{\x \sim \psi^k}[\x \in D^+] - 1 = 1.\]
On the other hand, for every $y \in \pmo^S$ such that  $\chi_S(y)
= -1$ and $D^- = (S, y)$, we have $\Pr_{\x \sim \psi^k}[\x  \in D^-]   = 0$  hence
    \[\skew_{\psi^k}(D^-) = 2^k \Pr_{\x \sim \psi^k}[\x \in D^-] - 1 = -1.\]
Since every parent of $D^+$ has $0$ skew, every such $D^+$ is a
$(1,1)$-minimal skewed subcube, and similary for every $D^-$.

Trivially, we have $\inorm{\psi^k} \leq 2^n$, hence
$\log(\inorm{\psi^k}) \leq n$ (in fact it equals $n -O(\log(n))$. Since the number of
  minimal weight codewords is $\Omega_k(n^{k/2 +1})$ by Lemma \ref{lem:sergey}, and $\gamma,
  \eps$ are both $1$, the number of codewords is
  $\Omega_k((\log(\inorm{\psi^k}))^{k/2 +1})$.
  Hence the number of minimal skewed subcubes is as claimed.
\end{proof}

\section{Algorithms for Finding Skewed Subcubes}
\label{sec:algo}

In this section, we present an algorithm that find skewed subcubes
efficiently in the random sample model, where we have access to random
samples from $\psi$.

To make \cref{alg:find_skew_recursive} efficient, we need to replace the step of
guessing $S$ (\cref{line:guess_s} in \cref{alg:find_skew_recursive}, and \cref{line:guess_s_neg} in \cref{alg:find_skew_neg}) with an algorithm to find
large low degree Fourier coefficients\footnote{We note that technically to implement the algorithm we also need to estimate $\ip{\psi, \mu_C}$ for every $C$ to an additive accuracy of $\min(\gamma, 2^{-k})$. If done via sampling, these only incur $2^{2k} k \log n / \gamma^2$ additional cost per call, and will be absorbed into our runtime bounds anyway.}. We restate the problem
below:

\begin{problem}
  {\bf Finding large low-degree biases.}
  Given a distribution $\psi$ on $\pmo^n$, an integer $k$ and $\rho \geq 0$, find
  all $S \subseteq [n]$ such that $|S| \leq k$ and
  \[ \hat{\psi}(S) \defeq \E_{\x \sim \psi}[\chi_S(x)] \geq \rho. \]
\end{problem}


Our main result is the following pair of theorems.
    \begin{restatable}[Algorithm for Positive Skew]{theorem}{algmain}
  \label{thm:alg_main}
	Given sample access to a distribution $\psi$ on $\pmo^n$,
        integer $k \leq n$, and parameters $\gamma \in(0,2^k - 1]$,
          $\eps \in(0,1]$ and $\lambda \in [0,1]$, there is an
            algorithm that returns all $(\gamma, \eps)$-minimal skewed
            subcubes of codimension at most $k$ in time:
            \begin{align*}
            \tilde O \left(n^{k\left(\frac{\omega}{3 - \lambda}\right)}\right)+  
            k^{O(k)} \cdot \left(\ln(e \inorm{\psi})\right)^k \left( \frac{\tilde O(n^{k/3})}{(\eps \sqrt{\gamma})^{4/\lambda}} + \frac{\poly(n)}{(\eps \sqrt{\gamma})^{2k}} \right)
            \end{align*}
        where $\omega$ is the matrix multiplication exponent, and $\tilde O$ hides $\poly \log n$ factors.
\end{restatable}

\begin{restatable}[Algorithm for Negative Skew]{theorem}{algmainneg}
	\label{thm:alg_main_neg}
	Given sample access to a distribution $\psi$ on $\pmo^n$,
	integer $k \leq n$, and parameters $\gamma \in(0,1]$,
	$\eps \in(0,1]$ and $\lambda \in [0,1]$, there is an
	algorithm that returns all $(-\gamma, \eps)$-minimal skewed
	subcubes of codimension at most $k$ in time:
	\begin{align*}
	\tilde O \left(n^{k\left(\frac{\omega}{3 - \lambda}\right)}\right)+  
	k^{O(k)} \cdot \left(\ln(e \inorm{\psi})\right)^k \left( \frac{\tilde O(n^{k/3})}{(\eps \gamma)^{4/\lambda}} + \frac{\poly(n)}{(\eps \gamma)^{2k}} \right)
	\end{align*}
	where $\omega$ is the matrix multiplication exponent, and $\tilde O$ hides $\poly \log n$ factors.
\end{restatable}

Theorem \ref{thm:intro_alg} follows from using $(a)$, setting $\lambda =
0.01$ and $\omega \leq 2.38$.

	In both algorithms, we will find large low-degree biases using a breakthrough algorithm of
        \cite{Valiant15} for detecting pairs of vectors that are highly correlated from
        a set of weakly correlated vectors. The algorithm was subsequently improved by \cite{KarppaKK18}).
	\begin{theorem}[\cite{KarppaKK18}]
		\label{thm:subquad_corr}
		Given two sets of vectors $V_1, V_2 \subseteq \pmo^n$
                for which there are at most $q$ pairs $(v_1, v_2) \in
                V_1 \times V_2$ with correlation larger than $\tau$,
                and a parameter $\rho \geq \tau^{1/\lambda}$ (for $\lambda \in [0,1]$), there is an
                algorithm $\textsc{FindCorr}(V_1, V_2, \rho, \tau)$
                that with high probability outputs all pairs $(v_1,
                v_2) \in V_1 \times V_2$ with correlation at least
                $\rho$. Furthermore,  algorithm runs in time
                \[\tilde O\left(n^{\frac{2 \omega}{3 - \lambda}} +
                qdn^{\frac{2(1 - \lambda)}{3 - \lambda}}\right).\]
	\end{theorem}

	The essence of the reduction is as follows. For each set $S \subseteq [n]$ less than $k/2$ we associate a vector $y_S$, for which each coordinate is a random sample of $\chi_S(x)$ where $x$ is drawn from $\psi$. If $Q, R \subseteq [n]$ are disjoint, the correlation coefficient $\expect{y_Q \cdot y_R}/d$ is precisely the value of the Fourier coefficient $\hatpsi(Q \cup R)$. Thus as long as the algorithm of \cite{KarppaKK18} succeeds and is not overwhelmed by sample error, every $\rho$ correlated pair $(y_Q, y_R)$ corresponds to a Fourier coefficient $\hatpsi(Q \cup R)$ of size less than $k$ and absolute value $\rho$. We now describe this more formally.

	\begin{algorithm}[ht]
		\caption{$\ffc(\psi, k, \rho, \lambda)$}
		\label{alg:find_fourier}
		\begin{algorithmic}[1]
			\State $\mathcal{S} \leftarrow \emptyset$
			\State $\tau \leftarrow (\rho/2)^{1/\lambda}$
			\State Draw a set of $d= O(k \log n /\tau^2)$ samples $x_1, \ldots x_d$  from $\psi$.
			\For{$T = O(k^{3/2} \log n)$ rounds}
			\State Randomly partition $[n]$ into two subsets $N_1$ and $N_2$.
			\State For every subset $S\subseteq N_1$ of size $\leq \lceil k/2\rceil$, form a vector $y_S \in [-1,1]^d$ for which the $i$th bit is set to $\chi_S(x_i)$. Call this set of vectors $V_1$.
			\State Do the same for $N_2$ for sets of size $\leq \lfloor k/2 \rfloor$, and call the set of vectors $V_2$.
			\State Run $\textsc{FindCorr}(V_1, V_2, \rho/2, \tau)$ from \cite{KarppaKK18} to find all pairs $y_Q$ and $y_R$ such that $Q \subseteq N_1$, $R \subseteq N_2$, and $y_Q$ and $y_R$ are $\rho/2$ correlated. For each of these, add $Q \cup R$ to $\mathcal{S}$. \label{line:corr_pairs}
			\EndFor
			\State \Return $\mathcal{S}$.
		\end{algorithmic}
	\end{algorithm}

	We first prove some simple lemmas using standard concentration of measure results.

        \eat{
		\begin{fact}[Hoeffding Bound, \cite{Hoeffding}]
			Let $X_1, \ldots, X_d \in [-1,1]$ be independent random variables such that for all $1 \leq i \leq d$, $\expect{X_i} = \mu$. Then
			\[\prob{\left|\frac{1}{d}\sum_{i=1}^d X_i - \mu \right| \geq \eps} \leq 2 e^{-d\eps^2/2}\]
		\end{fact}
}
		\begin{lemma}
			\label{lem:ip_hoeffding}
			For every $Q \in N_1$ and $R \in N_2$, w.h.p. $\left|\ip{y_Q, y_R} / d - \hatpsi(Q \cup R)\right| \leq \tau/2$.
		\end{lemma}
		\begin{proof}
			For a single $x \sim \pmo^n$ and disjoint $Q,
                        R \subseteq N$, we have $\expect{\chi_Q(x)
                          \cdot \chi_R(x)} = \hatpsi(Q \cup
                        R)$. Applying the Hoeffding bound with
                        our choice of $d= 32k \log n / \tau^2$, we
                        have that $\ip{y_Q, y_R}/d$ is within $\tau/2$
                        of $\hatpsi(Q \cup R)$ with probability at
                        least $1 - n^{-2k}$. By a union bound, this
                        hold for all $O(n^k)$ choices of pairs $(Q,
                        R)$ with probability $1-n^{-k}$.
		\end{proof}

		\begin{lemma}
			\label{lem:bisect}
			For every $S \subseteq[n]$ with $|S| \leq k$,
                        w.h.p. in at least one round $t \in [T]$ there
                        are $Q \subseteq N_1$, $R \subseteq N_2$ with
                        $Q \cap R = \emptyset$ such that $Q \cup R =
                        S$.
		\end{lemma}
		\begin{proof}
			Fix a set $S$ of size $\ell$. For a random
                        bipartition of $[n]$, the probability $S$ is
                        perfectly bisected is at least
                        $1/(8\sqrt{\ell}) \geq
                        1/(8\sqrt{k})$. The probability it is
                        never bisected over $T = 16k^{3/2} \log n$
                        rounds is upper bounded by
                        \[ \left(1 -
                        \frac{1}{8\sqrt{k}} \right)^T \leq e^{-2 k
                          \log n} = n^{-2k}.\]
                        By a union bound, every
                        $S$ of size $\leq k$ is bisected at least once
                        with high probability.
		\end{proof}

\begin{lemma}
\label{lem:ffc}
	The algorithm $\ffc(\psi, k, \rho, \lambda)$ returns all
        Fourier coefficients of $\psi$ of degree at most $k$ of
        absolute value at least $\rho$ in time 
        \[\tilde O \left( n^{k\omega/(3 - \lambda)} \right) +
        \tilde O(n^{k/3}) 2^{O(k)} (\ln(e \inorm{\psi}))^k\rho^{-4/\lambda}.\]
\end{lemma}
\begin{proof}
Consider any set $S \subseteq [n]$ of size $\leq k$ of magnitude at
least $\rho$. By \cref{lem:bisect}, w.h.p. for some round $t \in [T]$,
the algorithm will form two vectors $y_Q$ and $y_R$ such that $Q
\subseteq N_1$, $R \subseteq N_2$ and $Q \cup R = S$. Furthermore, by
\cref{lem:ip_hoeffding}, we have $\ip{y_Q, y_R}/d \geq \rho - \tau \geq \rho/2$.
In turn, this means that $\textsc{FindCorr}(V_1, V_2, \rho/2, \tau)$ will
detect these w.h.p.

To bound the running time, we need a bound on the number $q$ of pairs
with correlation higher than $\tau$. By \cref{lem:ip_hoeffding}, we have
$\ip{y_Q, y_R}/d \geq \tau$ implies $\hat{\psi}(Q \cup R) \geq
\tau/2$. The number of such coefficients is bounded by
\[ \frac{W^{\leq k}(\psi)}{(\tau/2)^2} \leq  2^8 \left(\ln(e\inorm{\psi})\right)^k\rho^{-2/\lambda}.\]
For each such coefficient $S$, there are $2^k$ ways to write it as $S
= Q \cup R$ for disjoint $Q, R$. Hence
\[ q \leq 2^{k + 8} (\ln(e \inorm{\psi}))^k\rho^{-2/\lambda}.\]

Observe that $\log_\tau\rho \geq \lambda$ by our choice of $\tau$.
By \cref{thm:subquad_corr}, \cite{KarppaKK18} will find a list
containing all $\rho/2$ correlated pairs in time at most
		\begin{align*}
		&\tilde O\left(\left(n^{k/2}\right)^{2\omega/(3 - \lambda)} +
                  qd\left(n^{k/2}\right)^{(2 - 2 \lambda) / (3-
                    \lambda)}\right)\\
                  & \leq \tilde O \left( n^{k\omega/(3 - \lambda)} +
                  2^{O(k)}(\ln(e \inorm{\psi}))^k\rho^{-2/\lambda}k \cdot \log(n)\cdot \rho^{-2/\lambda}n^{k(1 - \lambda)/(3 -
                    \lambda)}\right) \\
                  & \leq \tilde O \left( n^{k\omega/(3 - \lambda)} \right) +
                  \tilde O(n^{k/3}) 2^{O(k)} (\ln(e \inorm{\psi}))^k\rho^{-4/\lambda}. \qedhere
                  \end{align*}
\end{proof}

We relegate the remainder of the proofs of \cref{thm:alg_main} to the appendix.

\section{Reduction from Noisy Parity}
\label{sec:noisy-parity}

Recall that given $S \subseteq [n]$, a parity function $\chi_S: \pmo^n \rgta \pmo$ is given by
$\chi_S(x) = \prod_{i \in S}x_i$.
A noisy parity is a parity function with random noise of rate $\eta$
added to it. In other words, we say $f:\pmo^n \rightarrow \pmo$ is an $\eta$-noisy parity if
$\Pr[f(x) = \chi_S(x)]= 1-\eta$ an $\Pr[f(x) = -\chi_S(x)] =  \eta$.
In the sparse noisy parity problem, we are given access to  samples
$(\x,f(\x))$ where $\x \sim \mu$ is sampled uniformly from $\pmo^n$ and $f$ is
noisy parity $\chi_S$ with  $|S| = k$. The goal is to recover the
parity function, or equivalently the set $S$.

Given a set $S' \subseteq [n]$, we have
\[ \E_{\x \sim \mu}[\widehat{f}(S')] =
\begin{cases} 0 & \text{if~} S' \neq S \\
  1 - 2\eta & \text{if~} S' = S
\end{cases} \]
    This leads to a naive enumeration algorithm
that runs in time $O(n^k)$. The current best algorithm due to
\cite{Valiant15} runs in time $O(n^{0.8k})\poly(1/(1 - 2\eta)))$.
A series of reductions due to \cite{FeldmanGKP09} show that efficient
algorithms for sparse noisy parity  imply algorithms with similar
running times for learning $k$-juntas, decision trees and DNFs under
the uniform distribution. This suggests the following conjecture
(which we consider to be folklore).

\textbf{Conjecture:}
	There is no algorithm for the sparse noisy parity problem which runs in time $n^{o(k)}$.

We now prove \cref{thm:noisy-parity} which we restate below.
\noisyparity*
\begin{proof}
Finding a noisy parity reduces to finding a minimal skewed
subcube. Given an instance of sparse noisy parity, consider
the distribution $\psi$ on $\pmo^{n+1}$ obtained by appending the label to
the sample. Thus $\psi = (\x, f(\x))_{\x \sim \mu}$.
We show that all skewed subcubes must restrict the set $S \cup \{
n+1\}$, hence they are all minimal, and  finding any skewed subcube
solves the noisy parity problem.

For any $z \in \pmo^{S}$, consider the subcube $D^+(z)$ given by
$x_S =z$, $x_{n+1} = \chi_S(z)$. We have
\[ \skew(D^+(z)) = 2^{k+1}\left(\Pr_{\x' \sim \psi}[\x' \in D] -\Pr_{\x'
  \sim \mu}[\x' \in D]\right) = 2^{k+1}\left(\frac{1 -\eta}{2^k} -
\frac{1}{2^{k+1}}\right) =  1- 2\eta. \]
Similarly if we define $D^-(z)$ by $x_S = z$ and $x_{n+1} =
-\chi_S(z)$, then
\[ \skew(D^-(z)) = 2^{k+1}\left(\Pr_{\x' \sim \psi}[\x' \in D] -\Pr_{\x'
  \sim \mu}[\x' \in D]\right) = 2^{k+1}\left(\frac{\eta}{2^k} -
\frac{1}{2^{k+1}}\right) = -(1 - 2\eta). \]
It is easy to verify that if the set of restricted variables is not $S
\cup \{n+1\}$, then the skew is $0$, which shows that the subcubes
above all $(\pm(1- 2\eta), 1)$-minimal skewed subcubes.
\end{proof}

This shows the hardness of finding subcubes with skew $(1 - 2\eta) <
1$.  The reduction could be extended to show the hardness of finding
subcubes with larger skew, simply by concatenating $\ell$ different
samples of $\psi$. Now an algorithm that finds a subcube of skew $(2(1
- \eta))^\ell -1$ and codimension $k$ can be used to solve a $k/\ell$-sparse
noisy parity problem.

\section{The Membership Query Model}
\label{sec:mem-query}

\cref{thm:noisy-parity} suggests that a much better algorithm does not
exist in the model where we only get random samples from
$\psi$. However, noisy parity becomes trivial when we are
given query access to $f$, by repeatedly querying the function at $x$
and $x \cdot e_i$. This motivates us to consider the query model where
in addition to getting random samples from $\psi$,
we are allowed to query $\psi(x)$ for points $x$ of our choosing.
As we will see, this does make finding skewed subcubes easier
for distributions where $\inorm{\psi}$ is small. We first show how
this improvement arises, and then give evidence that queries do not
add too much power over random samples when $\inorm{\psi}$ is large.

Algorthmically, all we need is a procedure to find all large low-degree Fourier
coefficients of $\psi$ under the query model. Such a procedure is given by
a classic result~\cite{KM93} of Kushilevitz and Mansour, which uses the
algorithm of Goldreich and Levin~\cite{GL89}  to compute large Fourier
coefficients when given a  query access to a function.

\begin{theorem}[\cite{GL89}]
	\label{thm:GL}
	Given query access to $f:\pmo^n \rightarrow [-t,t]$ and a
        parameter $\rho > 0$, there is an algorithm running in time poly$(n,t/\rho)$ that with high probability outputs a list containing all subsets $S$, such that $\hatf(S) \geq \tau$.
\footnote{The theorem is typically stated for functions with range
$\pmo$, however a similar bound is true for the range $[-1,1]$. The
version stated here follows by scaling the function by $t$ so it lies in the
range $[-1,1]$, and replacing $\rho$ by $\rho/t$.
}
\end{theorem}

If we apply it to $\psi$, then we get an algorithm whose running time
is $\poly(n, \inorm{\psi}, \rho)$. Thus, the algorithm is faster than the trivial exhaustive search
algorithm only when  $\inorm{\psi} < n^{\alpha k}$ for some $\alpha >
0$. The polynomial dependence on $\inorm{\psi}$ in the running time is inevitable since
the algorithm finds all $\widehat{\psi}(S) \geq \rho$, and not just those
with $|S| \leq k$. The number of such coefficients can be
$\inorm{\psi}/\rho^2$. In contrast, when we restrict to $|S| \leq k$,
the list-size only grows as $\ln(e\inorm{\psi})^k/\rho^2$. This raises
the following natural open question:

\begin{problem}
	Given query access to a probability measure, $\psi$ such
        parameter $\rho > 0$, does there exist an algorithm that can
        find all $S$ such that $|S| \leq k$ and $|\widehat{\psi}(S)| \geq
        \rho$ in time $\poly(n, 1/\rho, \ln(\inorm{\psi})^k)$?
\end{problem}

We conclude by observing that some dependence on $\inorm{\psi}$  (at least
logarithmic) seems inevitable, even in cases where the list-size
is $1$. This is seen by a reduction from sparse noisy parity.
A  sample of size $O(k\log n/\epsilon^2)$ from an instance of noisy
parity preserves all correlations of sets of $k$ variables up to an
additive $\epsilon$.  Define $\psi$  to be the uniform measure of these samples alone.
Finding $S$ such that $|S| \leq k+1$ and $\widehat{\psi}(S) \geq 1 - 2\eta$ will solve the
noisy parity problem. Note that for $\psi$ the query model and the random samples model are
equivalent as we have the support explicitly. So if we believe that
sparse noisy parity requires time $n^{\Omega(k)}$ time, then any
algorithm for finding large low-degree Fourier coefficients in $\psi$
must require as much  time. Since $\inorm{\psi} =
\eps^22^n/k\log(n)$, we have $\ln(\inorm{\psi}) = n  -O(\log(1/\eps)$. This is consistent with a dependence of $\ln(\inorm{\psi})^{\Omega(k)}$.

{\footnotesize
\bibliography{refs}
\bibliographystyle{alpha}}
\section{Missing Proofs}
\label{sec:appendix}

\prodrule*
\begin{proof}
	We have
	\begin{align*}
	\ip{\psi, \mu_D} & = \ip{\psi, \mu_C} \cdot \ip{ \frac{\psi}{\ip{\psi, \mu_C}},
		\mu_D}
	= \ip{\psi, \mu_C}\cdot \ip{\restr{\psi}{C}, \mu_D \big|_C} \qedhere
	\end{align*}
	where the second equality follows from  $D \subseteq C$.
\end{proof}

\skewbyprob*
\begin{proof}
	We have
	\begin{align*}
	\ip{\psi, \mu_C} & = \E_{\x \in \mu_C}[\psi(\x)]\\
	& = \sum_{x \in C}\frac{\psi(x)}{2^{n-k}}\\
	& = 2^k \sum_{x \in C}\Pr_{\x \sim \psi}[\x =x]\\
	& = 2^k \Pr_{\x \sim \psi}[\x \in C].
	\end{align*}
        Hence
        \begin{align*}
	  \ip{\psi, \mu_C} - 1 &= 2^k \Pr_{\x \sim \psi}[\x \in C]  -1\\
          &= \frac{\Pr_{\x \sim \psi}[\x \in C] - 2^{-k}}{2^{-k}}.
         \end{align*}
	Since $\Pr_{\x \sim \mu}[\x \in C] = 2^{-k}$, the claim follows.
\end{proof}

\negposskew*
\begin{proof}
	Consider the sum:
	\begin{align*}
	\sum_{\substack{D = (K, w) \\ w \in \pmo^K}} \skew(D) =
	\sum_{\substack{D = (K, w) \\ w \in \pmo^K}} 2^k \left(\Pr_{x
		\sim \psi}[\x \in D] - \Pr_{x \sim \mu}[\x \in D] \right) =
	0.
	\end{align*}
	The first equality comes from \cref{lem:skew_by_prob}, the
	second follows since the set of cubes $D$ form a partition of
	$\pmo^n$.
\end{proof}

\skewavg*
\begin{proof}
	We have
	\begin{align*}
	\skew_\psi(C) & = 2^k\left(\Pr_{\x \sim \psi}[\x \in C ] - \frac{1}{2^k}\right)\\
	& = \frac{2^{k+ \ell}}{2^\ell}\left(\sum_{i=1}^{2^\ell}\left(\Pr_{\x \sim \psi}[\x \in C_i] - \frac{1}{2^{k+
			\ell}}\right)\right)\\
	& = \frac{1}{2^\ell}\left(\sum_{i=1}^{2^\ell}2^{k + \ell}\left(\Pr_{\x \sim \psi}[\x \in C_i] - \frac{1}{2^{k+
			\ell}}\right)\right)\\
	& = \frac{1}{2^\ell}\sum_{i=1}^{2^\ell} \skew_\psi(C_i). \qedhere
	\end{align*}
\end{proof}

\corlevel*
\begin{proof}
Recall that
\[ W^{\leq k}(\psi, J) = \sum_{T \subseteq J} \sum_{\substack{S
    \subseteq [n]\setminus J \\ |S| \leq k}} \hatpsi(S \cup T)^2 \]
  We will show that for any $T \subseteq J$, we have
  \begin{align}
    \label{eq:k-restrict}
    \sum_{\substack{S \subseteq [n]\setminus J\\|S| \leq k}}
    \hatpsi(S \cup T)^2 \leq e^2 (\ln(e \inorm{\psi}))^k.
    \end{align}
The claim will then follow by summing over all $2^{|J|}$ choices of
$T\subseteq J$.
To prove Equation \eqref{eq:k-restrict}, we define $\psi^{(T)}:
\pmo^{[n] \setminus J} \rgta \R$ as
\[ \psi^{(T)}(x) = \E_{\z \sim \pmo^J}[\psi(x \circ \z)\chi_T(\z)]. \]
Note that unlike $\psi$,  $\psi^{(T)}$ can be negative.
By orthogonality of characters, we have
\[ \psi^{(T)}(x) = \sum_{S \subseteq [n] \setminus J}\hatpsi(S \cup T)\chi_S(x). \]
Since $\psi^{(T)}$ is a signed average of $\psi$ which is non-negative,
  we have $\inorm{\psi^{T}} \leq \inorm{\psi}$ and
  \begin{align*}
    \pnorm{\psi^{T}}{1} & =  \E_{\x \in \pmo^{[n] \setminus
        J}}\left[ \left|\E_{\z \sim \pmo^J}[\psi(x \circ \z)\chi_T(\z)]
      \right|\right]\\
    & \leq \E_{\x \in \pmo^{[n] \setminus
        J}}\left[ \left|\E_{\z \sim \pmo^J}[|\psi(x \circ \z)|]
      \right|\right] = \E_{\x \in \pmo^n}[|\psi(\x)|] = 1.
  \end{align*}
The proof of Theorem \ref{thm:level_k} only uses bounds on the $1$ and
$\infty$ norms of $\psi$. Hence we can repeat the same proof with
$\psi^{(T)}$ to get an identical bound. This proves Equation
\eqref{eq:k-restrict}.
\end{proof}

\marginal*
\begin{proof}
  Let $|P| = p < n$. For $x \in \pmo^P$, we have
  \[ \psi_P(y) = 2^p\Pr_{\x \sim \psi}[\x_P = y] = \sum_{z \in
    \pmo^{\bar{P}}}\frac{\psi(y \circ z)}{2^{n-p}} = \E_{\z \sim
    \mu_{\bar{P}}}[\psi(y \circ \z)].\]
where the last expectation is over the bits $\z$ assigned to $\bar{P}$
being chosen uniformly at random. Using the Fourier expansion of $\psi$,
\[ \E_{\z \sim \mu_{\bar{P}}}[\psi(y \circ \z)] = \sum_{S
  \subseteq [n]}\hatpsi(S)\E_{\z \sim \mu_{\bar{P}}} \chi_S(y
  \circ \z) = \sum_{S \subseteq P}\hatpsi(S) \chi_S(y).\]
\end{proof}

\extension*
\begin{proof}
  It follows that $\psi$ is a distribution since it is non-negative
  and $\pnorm{\psi}{1} = 1$.
  Since the uniform distribution on $\pmo^{(\bar{P})}$ is given by
  $\mu_{\bar{P}}(y) = 1$ for all $y
  \in \pmo^{\bar{P}}$, it follows that $\psi = \psi' \times
  \mu_{\bar{P}}$ is the product of $\psi'$ and $\mu_{\bar{P}}$.

  Claim (2) follows trivially from the definition of $\psi$.

  For claim (3), we show that if $C = (K, y)$ is a minimal
  skewed subcube under $\psi$, then $K \subseteq P$. Indeed, suppose $i
  \in K \setminus P$. Consider the parent subcube $D \supsetneq C$ obtained by
      {\em freeing} the coordinate $i$. Then
      $\Pr_{\x \in \psi}[\x \in C] = \Pr_{\x \in \psi}[\x \in D]/2$,
      whereas $\codim(C) = \codim(D) + 1$, hence
      \[ \skew_\psi(C) = 2^{\codim(C)}\Pr_{\x \in \psi}[\x \in C] -1 =
        \skew_\psi(D).\]
      This violates the definition of minimality. In the other
      direction, it is easy to see that a minimal skewed subcube under
      $\psi'$ is also a minimal skewed subcube under $\psi$.
\end{proof}

\restriction*
\begin{proof}
Recall that $\restr{\psi}{C}(x) = \psi(x)/\ip{\psi, \mu_C}$ for $x \in
C$. For $x \in \pmo^{[n]\setminus J}$ let $x \circ z$ denote the
string obtained  by setting coordinates in $J$ to $z$ and those in
$[n]\setminus J$ to $x$. Then
\begin{align*}
  \restr{\psi}{C}(x) &= \frac{\psi(x \circ z)}{\ip{\psi, \mu_C}}\\
    & =  \frac{1}{\ip{\psi, \mu_C}}\sum_{S \subset [n] \setminus J}\sum_{T \subset
      J}\chi_{S \cup T}(x \circ z)\hatpsi(S \cup T)\\
  & =  \frac{1}{\ip{\psi, \mu_C}}\sum_{S \subset [n] \setminus J}\sum_{T \subset
      J}\chi_S(x)\chi_T(z)\hatpsi(S \cup T)\\
      & = \sum_{S \subset [n] \setminus J}\chi_S(x)\frac{\sum_{T \subset
      J}\chi_T(z)\hatpsi(S \cup T)}{\ip{\psi, \mu_C}}
\end{align*}
\end{proof}

\lemsergey*

\begin{proof}
  Let $d  =2e +2$ for $e \geq 0$.
  We use the fact that BCH codes are $[n, n - el - 1, 2e+2)]$ codes
  \cite[Theorem 16.21]{AlonS}. We need to show that there are many
  minimum weight codewords. For this, let us consider the parity check
  matrix $H$ which has dimension $(el + 1) \times n$ so that
  $\C_{BCH} = \{x \in \mathbb{F}_2^n: Hx = 0\}$.

Now let us consider the mapping $x \rightarrow Hx$ for all $x: \wt(x)
= e +1$. This maps each $x$ to a vector $y \in \zo^{el
  +1}$. For each $y \in \zo^{el + 1}$, let $b_y = |\{x:\wt(x) = e + 1, Hx = 0\}|$ be
the number of vectors of weight $e + 1$ mapped to $y$. Then we have
$\sum_y b_y = \binom{n}{e+1}$, hence
\[ \sum_y b_y^2 \geq \frac{\left(\sum_y b_y\right)^2}{2^{el + 1}} =
  \frac{\binom{n}{e+1}^2}{2(n+1)^e}.\]

For $x_1 \neq x_2$ both of weight $e +1$, if $Hx_1 =
Hx_2$ then $H(x_1 + x_2) = 0$, hence $x_1 + x_2$ is a non-zero codeword of
weight at most $2e +2$, hence it is in fact a minimum weight
codeword. Since there are $\binom{2e +2}{e+ 1} < 2^{2e+ 2}$ ways to write each
vector of weight $2e + 2$ as such as sum, hence the number of vectors
of codewords of weight $2e +2$ is at least
\[ \frac{1}{2^{2e +2}} \sum_y\binom{b_y}{2} = \frac{1}{2^{2e +3}}
    \sum_y(b_y^2 - b_y) \geq \frac{1}{2^{2e +3}}
    \left(\frac{\binom{n}{e+1}^2}{2(n+1)^e} - \binom{n}{e+1}\right) =
    \Omega_e(n^{e + 2}).\]
\end{proof}

\algmain*
\algmainneg*

The algorithm is $\fsr(\emptyset, \emptyset)$ in the positive case and $\fsn(\emptyset, \emptyset)$ in the negative case with the nondeterminism replaced. In the positive case, one could replace the enumeration of Fourier coefficients (this is \cref{line:guess_s} in \cref{alg:find_skew_recursive} and \cref{line:guess_s_neg} in \cref{alg:find_skew_neg}) by a call to $\ffc$. However this would naively yield a running time bound of $O(n^{0.8k}) \cdot k^{O(k)} \left(\ln (e \inorm{\psi} / \eps \gamma)\right)^k / (\eps \gamma)^{2k + O(1)}$. We show the stronger bound claimed in the theorem by making the following modification. Instead of running $\ffc$ at every recursive call, we run it once at the top level to get the list $L$ of heavy coefficients for $\psi$, and `deduce' the heavy Fourier coefficients for each restricted distribution $\restr{\psi}{C}$ from the original list $L$. To do so efficiently, we require a data structure which we now explain. 

We preprocess $L$ by creating a graph $G_L$. Vertices of this graph are indexed by elements of the power set $2^{[n]}$. For each coefficient $S \in L$, and each subset $T \subset S$, add the directed edge $T \rightarrow S$. Furthermore, each $T$ stores $k$ lists, where the $i^{th}$ list contains all sets $S$ in the out-neighborhood of $T$ such that $|S \backslash T| = i$. Since $2^k|L|$ is a bound on both the total number of edges and the total number of vertices in the graph, creating the graph takes $O(2^k |L|)$ time. Creating the partitions of the out-neighborhoods also takes $O(2^k |L|)$ times since each edge in the graph need only be processed once.

We summarize the algorithm below for reference.
\begin{algorithm}[ht]
	\caption{$\pc(L)$}
	\label{alg:preprocess}
	\begin{algorithmic}[1]
		\State $V, E \leftarrow \emptyset$
		\For{$S \in L$}
		\For{$T \subseteq S$}
		\State $V \leftarrow V \cup \{S, T\}$
		\State $E \leftarrow E \cup \{(T \rightarrow S)\}$
		\EndFor
		\EndFor
		\For{edge $(T \rightarrow S) \in E$}
		\State Let $i = |S \backslash T|$, and add $S$ to the $i^{th}$ list stored at $T$.
		\EndFor
		\State \Return $G_L = (V, E)$
	\end{algorithmic}
\end{algorithm}

Next, given $G_L$ the preprocessed form of $L$, a target subcube $C = (J,z)$ and a threshold $\tau$, we show how to output all Fourier coefficients $\widehat{\restr{\psi}{C}}(S)$ such that
$\widehat{\restr{\psi}{C}}(S) \geq \tau$.

\begin{algorithm}[ht]
\caption{$\dsc(G_L, C, \tau)$}
\label{alg:deduce_subcube_coeffs}
\begin{algorithmic}[1]
	\State Let $J$ be the coordinates fixed by $C$.
	\State $L' \leftarrow \emptyset$
	\For{$T \subseteq J$}
	\For{$S$ out-neighbor of $T$ in $G_L$ with $|S \backslash T| \leq k - |J|$}
	\State Check by sampling if $\left|\widehat{\restr{\psi}{C}}(S\backslash J)\right| \geq 3\tau/4$ to accuracy $\tau/4$. If so, add $S\backslash J$ to $L'$.
	\EndFor
	\EndFor
	\State \Return $L'$.
\end{algorithmic}
\end{algorithm}

\begin{lemma}
	\label{lem:deducecoeffs}
	Let $L = \{ S \subseteq [n]: |S| \leq k, |\hatpsi(S)|
	\geq   \tau/4^k\}$, and $G_L$ be the output of $\textsc{Preprocess}(L)$. Then $\dsc(G_L,C,\tau)$ returns the list of
	Fourier coefficients of $\restr{\psi}{C}$ of degree at
	most $k$ and magnitude at least $\tau$. Furthermore, $\dsc(G_L,C,\tau)$ runs in time 
	\[\poly(n) \cdot O(|L'|) \leq \poly(n) \cdot  2^{O(k)} \cdot \frac{(\ln(e \inorm{\psi}))^{k - |J|} }{\tau^2}.\]
\end{lemma}
\begin{proof}
	
	The output $L'$ consists only of sets $S'$ such that $S' \cup J \in L$, $S' \cap J = \emptyset$ and $|S'| \leq k - |J|$. Furthermore it contains all $S'$ meeting these criteria such that $\left |\widehat{\restr{\psi}{C}}(S')\right| \geq \tau$, note that for any set $H \in L$ such that $|H| \leq k$, if $T$ is chosen to be $H \cap J$, then $|H \backslash J| = |H \backslash T| \leq k - |J|$ and thus the algorithm will add $ H \backslash J$ to the output.
	
	To obtain the claim of the lemma, we need to argue that for every Fourier coefficient $\widehat{\restr{\psi}{C}}(S)$ of absolute value at least $\tau$, the set $S$ must appear in $L'$. We need the following consequence of Lemma
	\ref{lem:restrict}: let $\psi$ be a distribution on $\pmo^n$
	and let $C = (J,z)$ be a subcube. Then
	\[\widehat{\restr{\psi}{C}}(S) = \sum_{T \subseteq J}
	\frac{\chi_{T}(z) \hatpsi(S\cup T)}{\ip{\psi, \mu_C}}.\]
	It follows that every coefficient
	$\widehat{\restr{\psi}{C}}(S)$ is the signed sum of at
	most $2^k$ coefficients $\hatpsi(R) = \hatpsi(S \cup T)$, which is then
	scaled by $1 / \ip{\psi, \mu_C}$ (at most $2^k$). Furthermore, if $|S| \leq k - |J|$, then
	each such coefficient has $|R| \leq |S| + |T| \leq k$. Thus
	if $\widehat{\restr{\psi}{C}}(S) \geq \tau$, there is
	at least one $R \subseteq J$ with $|R| \leq k$ such that $|\widehat{\psi}(R)| \geq \tau/4^k$, which means that $R \in L$.
	
	Finally, the running time claim follows from \cref{cor:level_k} and the fact that for every $S \in L$ we have $\hatpsi(S) \geq \tau / 4^k$.
\end{proof}

%
%
%
%

We are now ready to prove our main theorem. We start with the positive case.

\begin{proof}[Proof of \cref{thm:alg_main}]
	We start by running $\ffc(\psi, k, \rho^+, \lambda)$ once, with $\rho^+ = \eps\sqrt{\gamma} /16^k$. This outputs a list $L^+$ containing all $S$ where $|S| \leq k$ and
	$\hatpsi(S) \geq \eps\sqrt{\gamma} / 4^k$. Subsequently, we compute $G_{L^+} \leftarrow \pc(L^+)$ using the output. This set up phase has running time $R^+$ bounded by
	\[R^+ \leq \tilde O \left(n^{k\left(\frac{\omega}{3 -
			\lambda}\right)}\right) + \frac{k^{O(k)} \tilde O\left (n^{k/3} \right)\ln(e
		\inorm{\psi})^k}{(\eps\sqrt{\gamma})^{4/\lambda}}. \]
	
	
	Next we run $\fsr(\emptyset, \emptyset)$ but we replace the nondeterministic enumeration of Fourier coefficients (this is \cref{line:guess_s} in \cref{alg:find_skew_recursive}) by the call $\dsc(G_{L^+}, C, \tau^+)$ where $\tau^+ := \eps \sqrt{\gamma} / \left(k_t \cdot \binom{k_t}{|S_t|}\right)$. We also replace the nondeterministic choice of $z$ (\cref{line:guess_z} in \cref{alg:find_skew_recursive}) by simple enumeration over all possible choices. Correctness follows from \cref{lem:complete} together with \cref{lem:ffc,lem:deducecoeffs} and it remains to show the running time bound.
	
	By \cref{lem:deducecoeffs}, at every subcube $C$ the algorithm spends time at most
	\[\poly(n) \cdot 2^{O(k)} \cdot \frac{(\ln(e \inorm{\psi}))^{k - |J|}}{(\tau^+)^2}  = \poly(n) \cdot 2^{O(k)} \cdot \frac{(\ln(e \inorm{\psi}))^{k - |J|}}{(\eps \sqrt{\gamma})^2} \] 
	to run $\dsc$. 
	
	On the other hand, the proof of \cref{lem:comb_bound} requires 
that the branching factor of $\fsr$ be bounded as in \cref{lem:guess_s}. 
Since this bound is at least $(\ln(e \inorm{\psi}))^{k-|J|} / (\eps \sqrt{\gamma})$ 
(and we may assume WLOG that the branching factor is \textit{at least} this threshold), 
we may amortize the cost of each call to $\dsc$ by charging to each child call the average running time per child. 
The time spent per child is $\poly(n) \cdot 2^{O(k)}$, and we argued in \cref{lem:comb_bound} that the total number of recursive calls to $\fsr$ is at most
	\[k^{O(k)} \left(\frac{\ln(e \inorm{\psi})}{\eps^2\gamma}\right)^k.\] 
	Thus we may bound the running time of $\fsr(\emptyset, \emptyset)$ by this expression as well.
	
	To conclude, the final running time of the algorithm in the positive skew case is:
	\begin{align*}
		&R^+ + \poly(n, k^k) \cdot k^{O(k)} \left(\frac{\ln(e \inorm{\psi})}{\eps^2\gamma}\right)^k \\
		&\leq \tilde O \left(n^{k\left(\frac{\omega}{3 - \lambda}\right)}\right) + \frac{k^{O(k)} \tilde O \left (n^{k/3} \right )\ln(e \inorm{\psi})^k}{(\eps\sqrt{\gamma})^{4/\lambda}} +\poly(n, k^k) \cdot k^{O(k)} \left(\frac{\ln(e \inorm{\psi})}{\eps^2\gamma}\right)^k \\
		&\leq \tilde O \left(n^{k\left(\frac{\omega}{3 - \lambda}\right)}\right)+  
		k^{O(k)} \cdot \left(\ln(e \inorm{\psi})\right)^k \left( \frac{\tilde O\left(n^{k/3}\right)}{(\eps \sqrt{\gamma})^{4/\lambda}} + \frac{\poly(n)}{(\eps \sqrt{\gamma})^{2k}} \right). \qedhere
	\end{align*}
\end{proof}

The negative case is identical, but with a different setting of parameters.

\begin{proof}[Proof of \cref{thm:alg_main_neg}]
In this case we run $\ffc(\psi, k, \rho^-, \lambda)$ with $\rho^- = \eps \gamma / 16^k$. This outputs the list $L^-$, and we set $G_{L^-} \leftarrow \pc(L^-)$. This all has running time $R^-$ bounded by
\[R^- \leq \tilde O\left(n^{k\left(\frac{\omega}{3 - \lambda}\right)}\right) + \frac{k^{O(k)}\tilde O\left(n^{k/3}\right)\ln(e\inorm{\psi})^k}{(\eps\gamma)^{4/\lambda}}.\]
Now we run $\fsn(\emptyset, \emptyset)$, with Fourier coefficient enumeration (\cref{line:guess_s_neg} in \cref{alg:find_skew_neg}) replaced with $\dsc(G_{L^-}, C, \tau^-)$, and we set $\tau^- = \eps \gamma / \left(k_t \cdot \binom{k_t}{|S_t|}\right)$. The analysis is identical to the positive case, and the final running time is:
\[\tilde O \left(n^{k\left(\frac{\omega}{3 - \lambda}\right)}\right)+  
k^{O(k)} \cdot \left(\ln(e \inorm{\psi})\right)^k \left( \frac{\tilde O\left(n^{k/3}\right)}{(\eps \gamma)^{4/\lambda}} + \frac{\poly(n)}{(\eps \gamma)^{2k}} \right). \qedhere\]
\end{proof}

\end{document}